\documentclass{sig-alternate}
\usepackage{xspace,colortbl}
\usepackage{url}
\usepackage[markup=nocolor]{changes}
\usepackage[english]{babel}
\usepackage{multirow}
\usepackage{subfigure}
\usepackage{graphicx}
\usepackage{amssymb}
\usepackage{fmtcount}
\usepackage{amsfonts}
\usepackage{xspace}
\usepackage{amsmath}
\usepackage{multirow}
\usepackage[mathscr]{eucal}
\usepackage{colortbl}
\usepackage{bm}
\usepackage[nospace]{cite}
\makeatletter
\newif\if@restonecol
\makeatother

\newcommand\blthanks[1]{%
  \begingroup
  \renewcommand\thefootnote{}\thanks{#1}
  \endgroup
}
\DeclareTextFontCommand{\emph}{\it}
\DeclareTextFontCommand{\em}{\it}

\usepackage[lined,boxed,vlined,ruled]{algorithm2e}

\long\def\comment#1{}
\begin{document}
\conferenceinfo{SIGMOD'13,} {June 22--27, 2013, New York, New York, USA.}
\CopyrightYear{2013}
\crdata{978-1-4503-2037-5/13/06}
\clubpenalty=10000
\widowpenalty = 10000

\tolerance=1000


\title{Leveraging Transitive Relations for Crowdsourced Joins*\blthanks{Revised September 2014. This is a revised and corrected version of a paper that appeared in the ACM SIGMOD 2013 Conference. The original version contained a claim that the algorithm for ordering pairs of records for presentation to the crowd (Section 4.2) was optimal. A subsequent paper in VLDB 2014 by Vesdapunt et al.~\cite{DBLP:journals/pvldb/NorasesBD14} showed that the ordering problem is, in fact, NP-hard. Thus, in this version we have removed the claim of optimality and have updated the discussion and example in Section 4.2 accordingly.  A more detailed explanation of this can be found in~\cite{2014arXiv1408.6916W}. The copyright notice on the SIGMOD 2013 paper is as follows:}}



\newtheorem{theorem}{Theorem}
\newtheorem{example}{Example}
\newtheorem{definition}{Definition}
\newtheorem{proposition}{Proposition}
\newtheorem{lemma}{Lemma}
\newtheorem{corollary}{Corollary}
\newcommand{\crowdert}{$\textsf{CrowdEST}$\xspace}

\newcommand{\dataset}{data set\xspace}
\newcommand{\datasets}{data sets\xspace}
\newcommand{\rp}{record pair\xspace}
\newcommand{\rps}{record pairs\xspace}
\newcommand{\order}{\omega\xspace}
\newcommand{\orderlabel}{\omega\ell\xspace}
\newcommand{\labelset}{\mathcal{L}\xspace}
\newcommand{\publish}{\mathcal{P}\xspace}
\newcommand{\clustergraph}{\textsc{ClusterGraph}\xspace}
\newcommand{\clustergraphs}{\textsc{ClusterGraphs}\xspace}
\newcommand{\cluster}[1]{\textsf{cluster}(#1)\xspace}
\newcommand{\clusterp}[1]{\textsf{cluster'}(#1)\xspace}
\newcommand{\cn}{\mathcal{C}\xspace}
\newcommand{\ecn}{\mathrm{E}\xspace}
\newcommand{\nphard}{$NP$-$hard$\xspace}
\newcommand{\paper}{{\textsf{Paper}}\xspace}
\newcommand{\cora}{{\textsf{Cora}}\xspace}
\newcommand{\transitive}{{\textsf{Transitive}}\xspace}
\newcommand{\nontransitive}{{\textsf{Non-Transitive}}\xspace}
\newcommand{\optimalorder}{{\textsf{Optimal Order}}\xspace}
\newcommand{\parallelno}{{\textsf{Parallel}}\xspace}
\newcommand{\nonparallel}{{\textsf{Non-Parallel}}\xspace}
\newcommand{\parallelid}{{\textsf{Parallel(ID)}}\xspace}
\newcommand{\parallelnf}{{\textsf{Parallel(ID+NF)}}\xspace}
\newcommand{\expectoptimalorder}{{\textsf{Expect Order}}\xspace}
\newcommand{\randomorder}{{\textsf{Random Order}}\xspace}
\newcommand{\worstorder}{{\textsf{Worst Order}}\xspace}
\newcommand{\product}{{\textsf{Product}}\xspace}
\newcommand{\abtbuy}{{\textsf{Abt-Buy}}\xspace}
\newcommand{\tp}{{\textsf{tp}}\xspace}
\newcommand{\precision}{{\textsf{precison}}\xspace}
\newcommand{\recall}{{\textsf{recall}}\xspace}
\newcommand{\fp}{{\textsf{fp}}\xspace}
\newcommand{\fn}{{\textsf{fn}}\xspace}
\newcommand{\prob}{\mathrm{P}\xspace}
\newcommand{\addtext}[1]{{{\textcolor{blue}{ #1}}\xspace}}
\pagestyle{plain}

\numberofauthors{1}
\author{\alignauthor Jiannan Wang{$\,^{\#}$},~~ Guoliang Li{$\,^{\#}$},~~ Tim Kraska{$\,^\dag$},~~ Michael J. Franklin{$\,^\ddag$},~~ Jianhua Feng{$\,^{\#}$} \\
\vspace{.2em}\affaddr{$^{\#}$Department of Computer Science,  Tsinghua University, ~~ $^\dag$Brown University, ~~ $^\ddag$AMPLab, UC Berkeley} \\
\vspace{.1em}\email{wjn08@mails.tsinghua.edu.cn,~~ligl@tsinghua.edu.cn,~~tim\_kraska@brown.edu}\\
\email{franklin@cs.berkeley.edu,~~~~fengjh@tsinghua.edu.cn}
}

\maketitle
\thispagestyle{plain}


\begin{abstract}
The development of crowdsourced query processing systems has recently attracted a significant attention in the database community. A variety of crowdsourced queries have been investigated. In this paper, we focus on the crowdsourced join query which aims to utilize humans to find all pairs of matching objects from two collections. As a human-only solution is expensive, we adopt a hybrid human-machine approach which first uses machines to generate a candidate set of matching pairs, and then asks humans to label the pairs in the candidate set as either matching or non-matching. Given the candidate pairs, existing approaches will publish all pairs for verification to a crowdsourcing platform. However, they neglect the fact that the pairs satisfy transitive relations. As an example, if $o_1$ matches with $o_2$, and $o_2$ matches with $o_3$, then we can deduce that $o_1$ matches with $o_3$ without needing to crowdsource $(o_1, o_3)$. To this end, we study how to leverage transitive relations for crowdsourced joins. We present a hybrid transitive-relations and crowdsourcing labeling framework which aims to crowdsource the minimum number of pairs to label all the candidate pairs. We propose a heuristic labeling order and devise a parallel labeling algorithm to efficiently crowdsource the pairs following the order. We evaluate our approaches in both simulated environment and a real crowdsourcing platform. Experimental results show that our approaches with transitive relations can save much more money and time than existing methods, with a little loss in the result quality.
\end{abstract}

\vspace{1mm}
\noindent
\sloppy
{\bf Categories and Subject Descriptors:} H.2.4 {[Database Management]}: Systems---\emph{Query processing};

\vspace{1mm}
\fussy

\section{Introduction}
\label{sec:intro}

\sloppy

The development of crowdsourced query processing systems has recently attracted a significant attention in the database community~\cite{conf/sigmod/FranklinKKRX11,journals/pvldb/MarcusWKMM11,ilprints1015}. A variety of crowdsourced queries have been investigated, such as crowdsourced \texttt{MAX} query~\cite{conf/sigmod/GuoPG12,conf/www/VenetisGHP12}, crowdsourced \texttt{SELECT} query~\cite{conf/icde/BethTMP13,ilprints1049}, and crowdsourced \texttt{JOIN} query~\cite{journals/pvldb/MarcusWKMM11,journals/pvldb/WangKFF12,conf/www/DemartiniDC12,ilprints1047}. In this paper, we focus on the crowdsourced join query for entity resolution which aims to identify all pairs of \emph{matching} objects between two collections of objects, where humans are utilized to decide whether a pair of objects is matching, i.e. referring to the same real-world entity. The crowdsourced join query can help to solve many real problems that are hard for computers. For example, given two collections of product records from two online retailers, there may be different records that refer to the same product, e.g. ``\textsf{iPad~2nd~Gen}" and ``\textsf{iPad~Two}". In order to integrate them, we can perform a crowdsourced join query to find all pairs of matching products.

\fussy

A human-only implementation of crowdsourced joins is to ask humans to label every pair of objects from the two collections as either matching or non-matching~\cite{journals/pvldb/MarcusWKMM11}. Since the human-only solution is wasteful, prior works~\cite{journals/pvldb/WangKFF12,conf/www/DemartiniDC12,ilprints1047} showed how to build hybrid human-machine approaches. In our paper, we adopt a hybrid approach which first uses machines to generate a candidate set of matching pairs, and only then asks humans to label the pairs in the candidate set. Given the candidate pairs, existing hybrid solutions will publish all of them to a crowdsourcing platform, e.g., Amazon Mechanical Turk (AMT). However, existing hybrid solutions neglect the fact that the pairs satisfy transitive relations. By applying transitive relations, we can deduce some pairs' labels without asking humans to label them, thus reducing the number of crowdsourced pairs. For example, if $o_1$ and $o_2$ are matching, and $o_2$ and $o_3$ are matching, we do not need to crowdsource the label for $(o_1, o_3)$ since they can be deduced as matching based on transitive relations.


Based on this idea, in our paper, we study the problem of combining transitive relations and crowdsourcing to label the candidate pairs generated by machines. We formulate this problem, and propose a hybrid labeling framework which aims to crowdsource the minimum number of pairs for labeling all the candidate pairs. We find that the labeling order, i.e. which pairs should be labeled first, has a significant effect on the total number of crowdsourced pairs. We prove that labeling first all matching pairs, and then the other non-matching pairs leads to the optimal labeling order. However, the optimal order requires to know the real matching pairs upfront which cannot be achieved in reality. Therefore, we propose a heuristic labeling order which labels the pairs in the decreasing order of the likelihood that they are a matching pair. The likelihood could be given by some machine-learning methods~\cite{journals/pvldb/WangKFF12}. In order to label the pairs in this order, one simple way is to label them from the first pair to the last pair one by one. However, this method prohibits workers from doing tasks in parallel and leads to a long completion time. To address this problem, we devise a parallel labeling algorithm which can identify the pairs that must need to be crowdsourced, and ask crowd workers to label them in parallel. To summarize, we make the following contributions in the paper:

\sloppy

\vspace{-.25em}
\begin{itemize}
  \item We formulate the problem of utilizing transitive relations to label the candidate pairs in crowdsourcing, and propose a hybrid transitive-relations and crowdsourcing labeling framework to address this problem.\vspace{-.5em}
  \item We find the labeling order has a significant effect on the number of crowdsourced pairs, and respectively propose an optimal labeling order and a heuristic labeling order. 
  \vspace{-.5em}
  \item We devise a parallel labeling algorithm to reduce the labeling time, and propose two optimization techniques to further enhance the performance.\vspace{-.5em}
  \item We present our evaluations using both simulation and AMT. The experimental results show that our approaches with transitive relations can save much more money and time than existing methods, with a little loss in the result quality.
\end{itemize}
\vspace{-.25em}

\fussy

{\noindent \bf Organization.} We formulate our problem in Section~\ref{sec:problem formulation} and propose a hybrid labeling framewo1rk in Section~\ref{sec:hybrid-workflow}. Section~\ref{sec:sorting} discusses the optimal and expected optimal labeling orders. We devise a parallel labeling algorithm and two optimization techniques in Section~\ref{sec:parallel}. Experimental study is presented in Section~\ref{sec:exp}. We cover related work in Section~\ref{sec:related-work}, and present our conclusion and future work in Section~\ref{sec:conclusion}.


\section{Problem Formulation}\label{sec:problem formulation}


In this section, we first introduce crowdsourcing (Section~\ref{subsec:crowdsourcing}) and then define transitive relations (Section~\ref{subsec:transitive-relations}). Finally, we formulate our problem of utilizing transitive relations to label a set of pairs in crowdsourcing (Section~\ref{subsec:problem-description}).

\subsection{Crowdsourcing}\label{subsec:crowdsourcing}

There are many crowdsourcing platforms, such as AMT and MobileWorks, which provide APIs for easily calling large numbers of workers to complete micro-tasks (called Human Intelligent Tasks (HITs)). To label whether two objects in a pair are identical through crowdsourcing, we create an HIT for the pair, and publish it to a crowdsourcing platform. Figure~\ref{fig:hit-example} shows an example HIT for a pair (``\textsf{iPad 2}", ``\textsf{iPad two}"). In the HIT, the workers are required to submit ``YES" if they think ``\textsf{iPad 2}" and ``\textsf{iPad two}" are the same or submit ``NO" if they think ``\textsf{iPad 2}" and ``\textsf{iPad two}" are different. After the workers have completed the HIT, we obtain the crowdsourced label of the pair.

\begin{figure}[tbp]
\centering
  \includegraphics[scale=0.35]{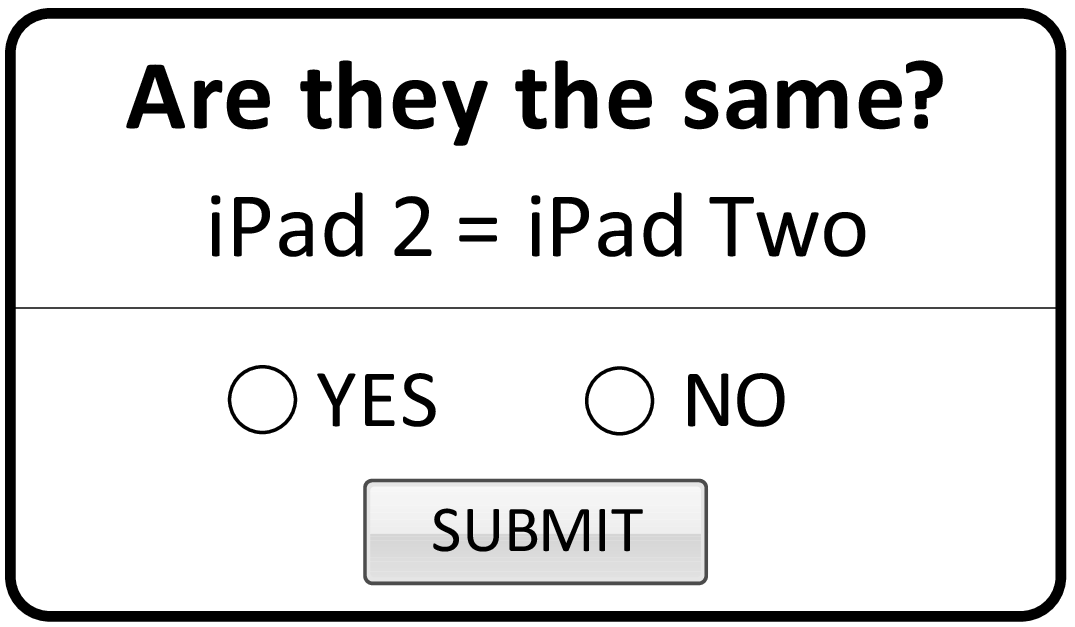}\\ \vspace{-0.5em}
  \caption{An example HIT for an object pair.}\label{fig:hit-example} \vspace{-1em}
\end{figure}

{\noindent \bf Assumption:} Of course, workers might return wrong results and ambiguities in the question might exist.
Techniques to address this problem were, for example, proposed in~\cite{Ipeirotis2010,conf/colt/DekelS09,wais2010towards, journals/pvldb/LiuLOSWZ12}. However, this problem can be treated as an orthogonal issue, as shown by~\cite{journals/pvldb/ParameswaranSGPW11,ilprints1047}, and we assume only correct answers for the remainder of the paper.



\subsection{Transitive Relations}\label{subsec:transitive-relations}

In this section, we discuss how to utilize transitive relations to deduce the label of a pair. We use the following notations. Let $p=(o, o')$ denote an object pair. The label of $p=(o, o')$  could be either ``matching" or ``non-matching", which respectively means that $o$ and $o'$ refer to the same real-world entity, and $o$ and $o'$ refer to different real-world entities. If $o$ and $o'$ are matching (non-matching), they are denoted by $o=o'$ ($o \neq o'$).

There are two types of transitive relations.

\vspace{.5em}

{\noindent {\bf Positive Transitive Relation:} Given three objects, $o_1$, $o_2$ and $o_3$, if $o_1=o_2$, and $o_2=o_3$, then we have $o_1=o_3$. For example, consider the three objects ``\textsf{iPad~2nd~Gen}", ``\textsf{iPad~Two}" and ``\textsf{iPad~2}". As ``\textsf{iPad~2nd~Gen}" and ``\textsf{iPad~Two}" are matching, and ``\textsf{iPad~Two}" and ``\textsf{iPad~2}" are matching, then ``\textsf{iPad~2nd~Gen}" and ``\textsf{iPad~2}" can be deduced as a matching pair based on positive transitive relation.

\vspace{.5em}

{\noindent {\bf Negative Transitive Relation:} Given three objects, $o_1$, $o_2$ and $o_3$, if $o_1=o_2$, and $o_2\neq o_3$, then we have $o_1\neq o_3$. For example, consider the three objects ``\textsf{iPad~Two}", ``\textsf{iPad~2}" and ``\textsf{iPad~3}". As ``\textsf{iPad~Two}" and ``\textsf{iPad~2}" are matching, and ``\textsf{iPad~2}" and ``\textsf{iPad~3}" are non-matching, then ``\textsf{iPad~Two}" and ``\textsf{iPad~3}" can be deduced as a non-matching pair based on negative transitive relation.

\vspace{.5em}

By applying positive and negative transitive relations to $n$ objects, we have the following lemma.
\begin{lemma}\label{lem:transitivity}
Given a set of objects, $o_1, o_2, \cdots, o_n$, (1) if $o_i = o_{i+1}$ $(1 \leq i < n)$, then we have $o_1 = o_{n}$; (2) if $o_i = o_{i+1}$ $(1 \leq i < n, i \neq k)$, and $o_k \neq o_{k+1}$, then we have $o_1 \neq o_{n}$.
\end{lemma}

Given a set of labeled pairs, to check if a new pair $(o, o')$ can be deduced from them, we build a graph for the labeled pairs for ease of presentation. In the graph, each vertex represents an object, and each edge denotes a labeled pair. (For simplicity, an object and a pair of objects are respectively mentioned interchangeably with its corresponding vertex and edge in later text.) From Lemma~\ref{lem:transitivity}, we can easily deduce the following conditions:
\begin{enumerate}
  \item If there exists a path from $o$ to $o'$ which only consists of matching pairs, then $(o, o')$ can be deduced as a matching pair;
  \item If there exists a path from $o$ to $o'$ which contains a single non-matching pair, then $(o, o')$ can be deduced as a non-matching pair;
  \item If any path from $o$ to $o'$ contains more than one non-matching pair, $(o, o')$ cannot be deduced.
\end{enumerate}

\begin{example}\label{exa:trans}
Consider seven labeled pairs: three matching pairs $(o_1, o_2), (o_3, o_4), (o_4, o_5)$ and four non-matching pairs $(o_1, o_6), (o_2, o_3), (o_3, o_7), (o_5, o_6)$. To check whether the unlabeled pairs $(o_3, o_5), (o_5, o_7), (o_1,o_7)$ can be deduced from them, we first build a graph as shown in Figure~\ref{fig:deduce-graph}.

For the unlabeled pair $(o_3, o_5)$, there is a path $o_3\!\rightarrow\!o_4\!\rightarrow\!o_5$ from $o_3$ to $o_5$ which only consists of matching pairs, i.e. $o_3 = o_4$, $o_4 = o_5$, thus $(o_3, o_5)$ can be deduced as a matching pair.

For the unlabeled pair $(o_5, o_7)$, there is a path $o_5\!\rightarrow\!o_4\!\rightarrow\!o_3\!\rightarrow\!o_7$ from $o_5$ to $o_7$ which contains a single non-matching pair, i.e. $o_5 = o_4$, $o_4 = o_3$, $o_3 \neq o_7$, thus $(o_5, o_7)$ can be deduced as a non-matching pair.

For the unlabeled pair $(o_1, o_7)$, there are two paths $o_1\!\rightarrow\!o_2\!\rightarrow\!o_3\!\rightarrow\!o_7$ and $o_1\!\rightarrow\!o_6\!\rightarrow\!o_5\!\rightarrow\!o_4\rightarrow\!o_3\rightarrow\!o_7$ from $o_1$ to $o_7$. As both of them contain more than one non-matching pair, $(o_1, o_7)$ cannot be deduced.
\end{example}

\begin{figure}[htbp]\vspace{-1em}
\centering
  \includegraphics[scale=0.35]{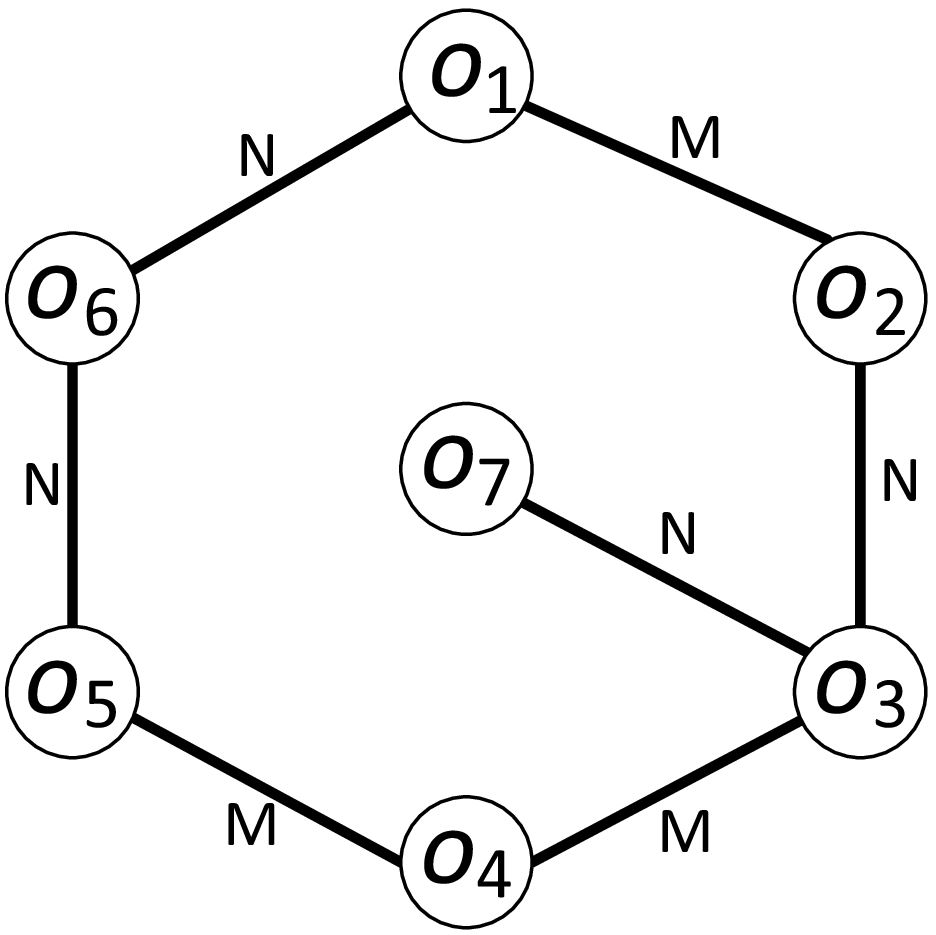}\\\vspace{-.5em}
  \caption{Graph illustrating Example~\ref{exa:trans} (``M" denotes ``matching", and ``N" denotes ``non-matching").}\label{fig:deduce-graph}\vspace{-1em}
\end{figure}

\subsection{Problem Description}\label{subsec:problem-description}
To process crowdsourced joins, we first use machine-based techniques to generate a candidate set of matching pairs. This has already been studied by previous work~\cite{journals/pvldb/WangKFF12}. The goal of our work is to study how to label the candidate pairs. Since in our setting, some pairs will be labeled by crowd workers, and others will be deduced using transitive relations.~We call the former \emph{crowdsourced (labeled) pairs}, and the latter \emph{deduced (labeled) pairs}. Typically, on a crowdsourcing platform, we need to pay for crowdsourced pairs, thus there is a financial incentive to minimize the number of crowdsourced pairs. Based on this idea, we define our problem as below.

\begin{definition}
Given a set of pairs that need to be labeled, our goal is to crowdsource the minimum number of pairs such that for the other pairs, their labels can be deduced from the crowdsourced pairs based on transitive relations.
\end{definition}

\begin{figure}[htbp]
\centering
  \includegraphics[scale=0.32]{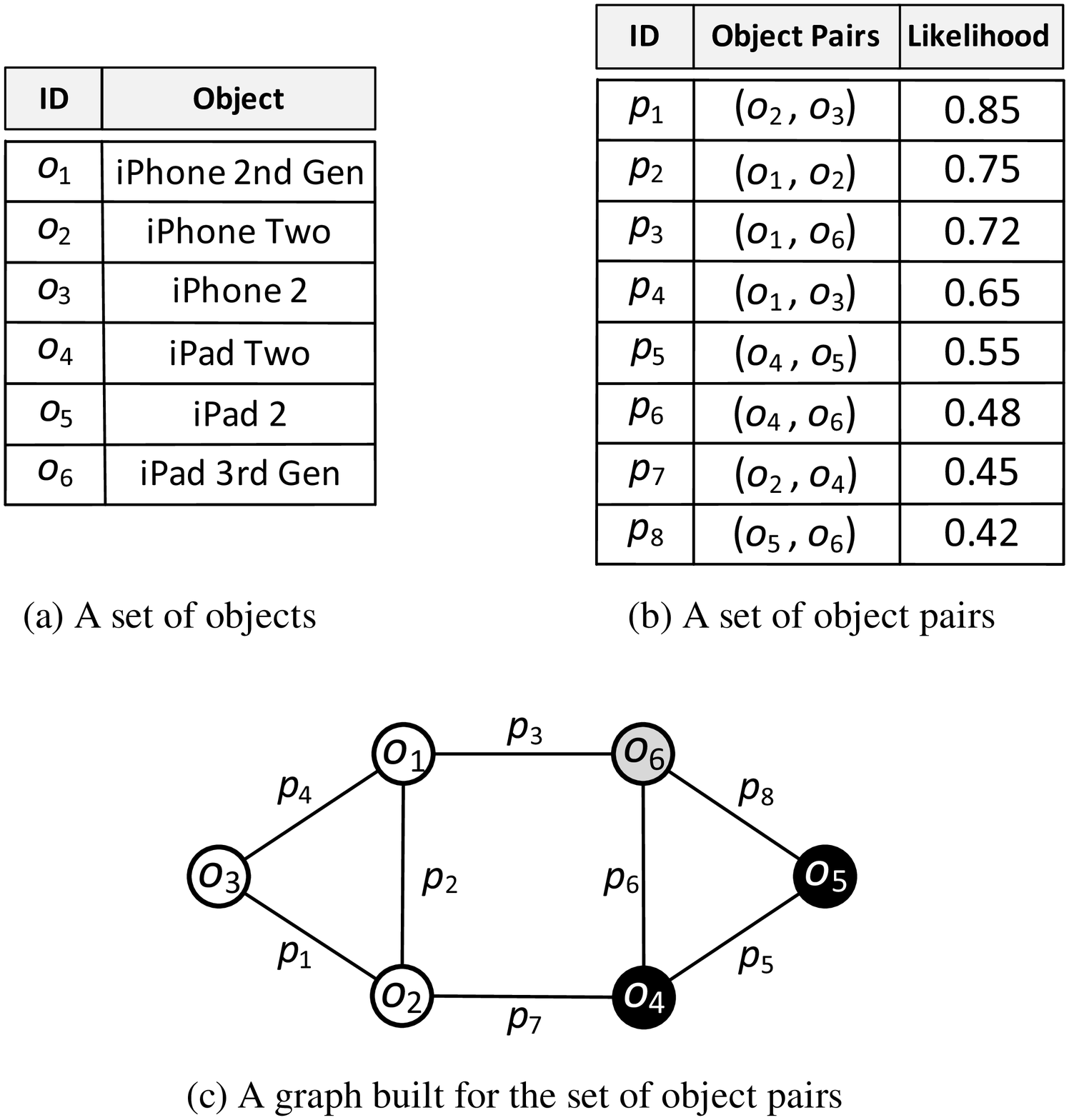}\\ 
  \caption{A running example.}\label{fig:running-example}\vspace{-.5em}
\end{figure}

\begin{example}\label{exa:problem}
Figure~\ref{fig:running-example} shows eight pairs, i.e., $p_1, p_2\cdots, p_8$, generated by machine-based methods for labeling (Please ignore the Likelihood column for now). We build a graph for these pairs, where the vertices with the same grey level represent the matching objects. One possible way to label them is to crowdsource seven pairs $p_1, p_2, p_3, p_5, p_6, p_7, p_8$. For the other pair $p_4$, as shown in the graph, it can be deduced from $p_1$ and $p_2$ based on transitive relations. A better way to label them only needs to crowdsource six pairs $p_1, p_2, p_3, p_5, p_7, p_8$. For the other pairs, as shown in the graph, $p_4$ can be deduced from $p_1$ and $p_2$, and $p_6$ can be deduced from $p_5$ and $p_8$. It is not possible to further reduce the amount of crowdsourced pairs. Thus, six is the optimal amount.
\end{example}

\section{Labeling Framework}
\label{sec:hybrid-workflow}

We propose a hybrid transitive-relations and crowdsourcing labeling framework in this section. Our framework takes as input a set of unlabeled pairs generated by machine-based techniques, and identifies these pairs' labels either through crowdsourcing or by using transitive relations. As shown in Figure~\ref{fig:workflow}, our framework mainly consists of two components, \texttt{Sorting} and \texttt{Labeling}. Their details will be described in Sections~\ref{subsec:sorting} and~\ref{subsec:labeling}, respectively.


\begin{figure}[htbp]
\centering
  \includegraphics[scale=0.3]{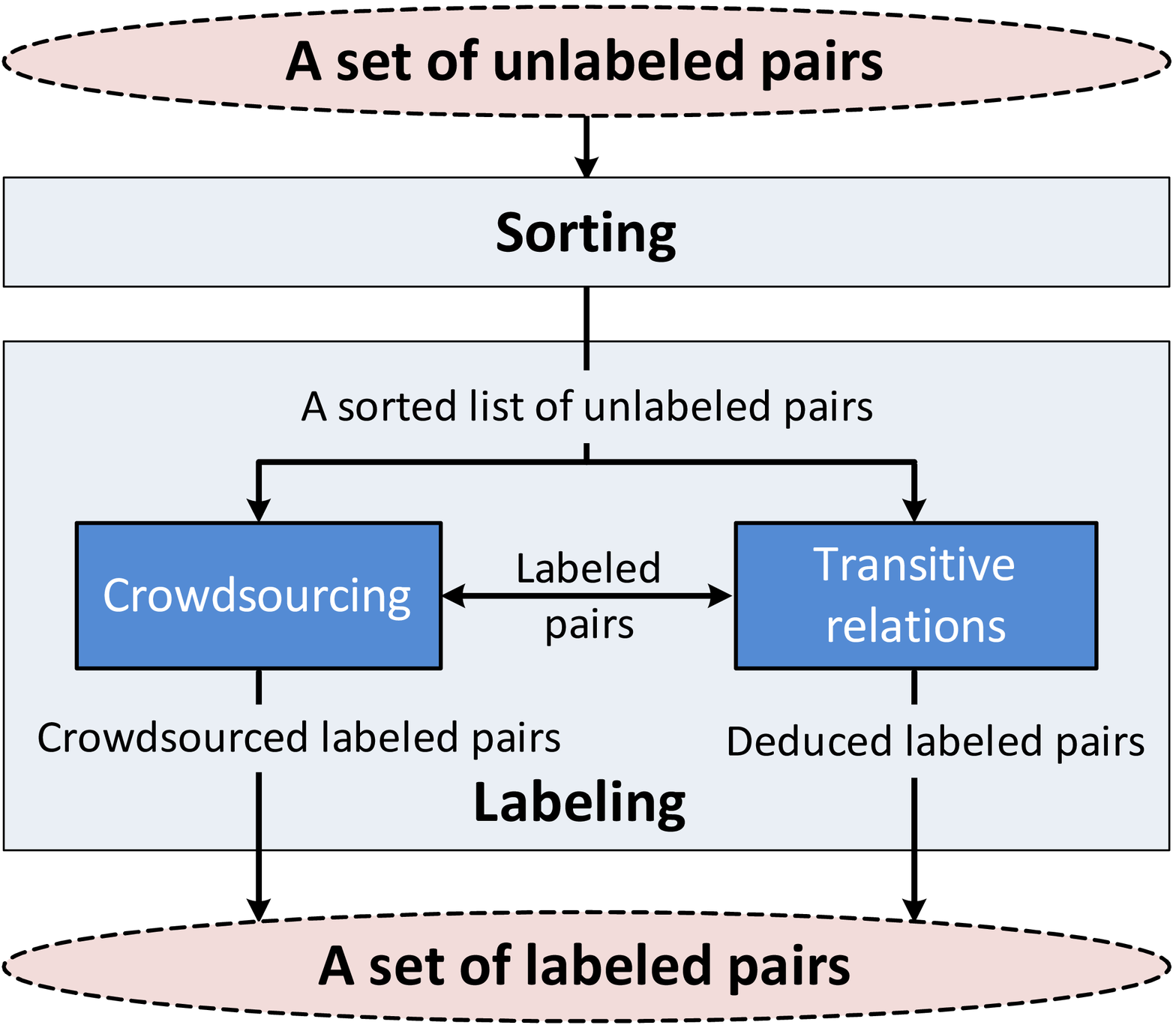}\\\vspace{-.5em}
  \caption{Hybrid transitive-relations and crowdsourcing labeling framework.}\label{fig:workflow} 
\end{figure}

\subsection{Sorting Component}\label{subsec:sorting}

Given a set of unlabeled pairs, we have an interesting finding that the labeling order of the pairs will affect the number of crowdsourced pairs. A labeling order can be taken as a sorted list of pairs, denoted by $\order=\langle p_1, p_2, \cdots, p_n \rangle$, where $p_i$ $(2\leq i \leq n)$ will be labeled after $p_1, p_2, \cdots, p_{i-1}$. For example, suppose we need to label three pairs, $(o_1, o_2)$, $(o_2, o_3)$, $(o_1, o_3)$, where $o_1=o_2$ and $o_2 \neq o_3$ and $o_1 \neq o_3$. If the labeling order is $\order = \big\langle (o_1, o_2), (o_2, o_3), (o_1, o_3)\big\rangle$, after labeling the first two pairs through crowdsourcing, we obtain $o_1=o_2$ and $o_2 \neq o_3$. For the third pair, we can deduce $o_1\neq o_3$ from $o_1=o_2$ and $o_2 \neq o_3$ based on transitive relations, thus $\order$ requires crowdsourcing \emph{two} pairs. However, if we use a different~labeling order $\order' = \big\langle (o_2, o_3), (o_1, o_3), (o_1, o_2)\big\rangle$, after labeling the first two pairs through crowdsourcing, we obtain $o_2 \neq o_3$ and $o_1\neq o_3$. We are unable to deduce $o_1 = o_2$ from $o_2 \neq o_3$ and $o_1\neq o_3$ based on transitive relations, thus $\order'$ requires crowdsourcing \emph{three} pairs which is more than that required by $\order$.

Based on this observation, in our framework, the sorting component attempts to identify the optimal labeling order to minimize the number of crowdsourced pairs. Thus, it takes as input a set of unlabeled pairs and outputs a sorted list of unlabeled pairs. The details of identifying the optimal labeling order will be presented in Section~\ref{sec:sorting}.

\subsection{Labeling Component}\label{subsec:labeling}



Given a sorted list of unlabeled pairs, the labeling component labels the pairs in the sorted order. In this section, we present a very simple, one-pair-at-a-time, labeling algorithm to achieve this goal. Consider a sorted list of pairs $\order = \langle p_1, p_2, \cdots, p_n\rangle$. The algorithm will start with labeling from the first pair, and then label each pair one by one. When labeling the $i$-th pair $p_i$, if its label cannot be deduced from the already labeled pairs (i.e., $\{p_1, p_2, \cdots, p_{i-1}\}$) based on transitive relations, we publish $p_i$ to a crowdsourcing platform and obtain its crowdsourced label; otherwise, we deduce its label from $p_1, p_2, \cdots, p_{i-1}$, and output the deduced label. After obtaining the label of $p_i$, we begin to process the next pair $p_{i+1}$, and use the same method to get its label. The algorithm stops until all the pairs are labeled.


\begin{figure}[tup]\vspace{-1em}
\begin{algorithm}[H]

\linesnumbered \SetVline

\small
\caption{\textsc{DeduceLabel}($p$, $\labelset$)}

\SetLine

\KwIn{$p = (o, o')$~: an object pair; \hspace{0.3em}$\labelset$: a set of labeled pairs}
 \KwOut{$\ell$: the deduced label}

\SetVline

\Begin{
    Build a $\clustergraph$ for $\labelset$;\nllabel{algo:deduce-label:buildgraph}\\
    Let $\cluster{o}$ and $\cluster{o'}$ denote the cluster of objects $o$ and $o'$ respectively\;
    \If {\emph{$\cluster{o} = \cluster{o'}$}}{\nllabel{algo:deduce-label:matching-begin}
         $\ell$ =  ``matching";\nllabel{algo:deduce-label:matching-end}\\
    }
    \Else{
        \If {\emph{there is an edge between $\cluster{o}$ and $\cluster{o'}$}}{\nllabel{algo:deduce-label:nonmatching-begin}
            $\ell$ =  ``non-matching";\nllabel{algo:deduce-label:nonmatching-end}\\
        }
        \Else{\nllabel{algo:deduce-label:undeduced-begin}
             $\ell$ =  ``undeduced";\nllabel{algo:deduce-label:undeduced-end}\\
        }
    }
    \textbf{return} $\ell$\;
}
\end{algorithm} \vspace{-1.5em}
\caption{\textsf{DeduceLabel} algorithm. }\label{algo:deduce-label}\vspace{-1em}
\end{figure}

Next, we discuss how to check whether an unlabeled pair $p = (o, o')$ can be deduced from a set of labeled pairs based on transitive relations. As mentioned in Section~\ref{subsec:transitive-relations}, we can build a graph for the labeled pairs, and check the graph whether there is a path from $o$ to $o'$ which contains no more than one non-matching pair. If there exists such a path, $p$ can be deduced from the labeled pairs; otherwise, $p$ cannot be deduced. One naive solution to do this checking is enumerating every path from $o$ to $o'$, and counting the number of non-matching pairs in each path. However, the number of enumerated paths may increase exponentially with the number of vertices in the graph, thus we propose an efficient graph-clustering-based method to solve this problem.


When enumerating each path, we find that only non-matching pairs in the path can affect the checking result. In other words, the matching pairs have no effect on the checking result. This observation inspires us to merge the matching objects into the same cluster, and then, for each pair of non-matching objects, we add an edge between their corresponding clusters. We call the new graph a \clustergraph. By using the \clustergraph, we can efficiently check whether an unlabeled pair $p=(o, o')$ can be deduced from the already labeled pairs. Figure~\ref{algo:deduce-label} shows the pseudo-code of the algorithm. Given a set of labeled pairs $\labelset$, we first build a \clustergraph for $\labelset$ using Union-Find algorithm~\cite{journals/jacm/Tarjan75} (Line~\ref{algo:deduce-label:buildgraph}). Then for the unlabeled pair $p=(o, o')$,

(1) If $o$ and $o'$ are in the same cluster, then there is a path from $o$ to $o'$ which only consists of matching pairs. Thus, $p$ can be deduced as a matching pair (Lines~\ref{algo:deduce-label:matching-begin}-\ref{algo:deduce-label:matching-end});

(2) If $o$ and $o'$ are in two different clusters,

(2.1) If there is an edge between the two clusters, there exists a path from $o$ to $o'$ with a single non-matching pair. Thus, $p$ can be deduced as a non-matching pair (Lines~\ref{algo:deduce-label:nonmatching-begin}-\ref{algo:deduce-label:nonmatching-end});

(2.2) If there is no edge between the two clusters, there does not exist a path from $o$ to $o'$ with no more than one non-matching pair. Thus, $p$ cannot be deduced (Lines~\ref{algo:deduce-label:undeduced-begin}-\ref{algo:deduce-label:undeduced-end}).

\begin{figure}[tbp]\vspace{-1em}
\centering
  \includegraphics[scale=0.35]{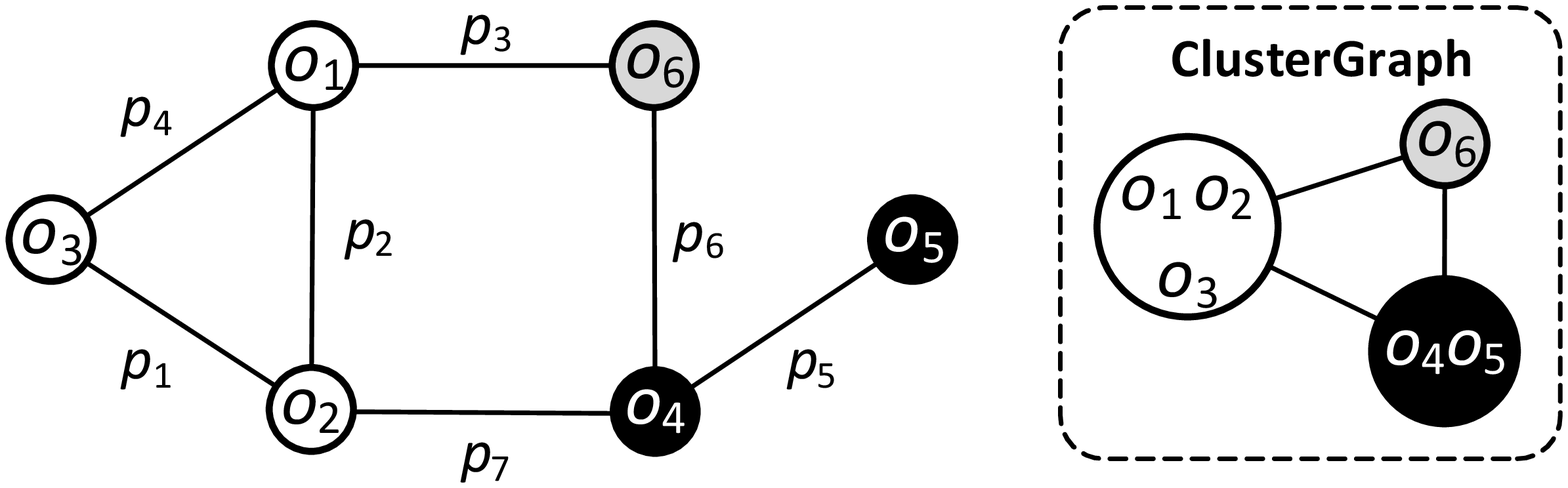}\\\vspace{-1em}
  \caption{A {$\clustergraph$} built for the first seven labeled pairs $\{p_1, p_2, \cdots, p_7\}$ in Figure~\ref{fig:running-example}.}\label{fig:cluster-graph}\vspace{-.5em}
\end{figure}

\begin{example}
Suppose we have already labeled seven pairs $\{ p_1, p_2, \cdots, p_7\}$ in Figure~\ref{fig:cluster-graph}. To check whether $p_8 = (o_5, o_6)$ can be deduced from them, we first build a $\clustergraph$ as follows. Since $o_1, o_2, o_3$ are matching, we merge them into one cluster. As $o_4$ and $o_5$ are matching, we merge them into another cluster. Since $o_6$ does not match with any other object, we take itself as one cluster. There are three non-matching pairs, i.e., $p_3=(o_1, o_6)$, $p_6=(o_4, o_6)$ and $p_7=(o_2, o_4)$. We respectively add three non-matching edges between $\cluster{o_1}$ and $\cluster{o_6}$, $\cluster{o_4}$ and $\cluster{o_6}$, $\cluster{o_2}$ and $\cluster{o_4}$.

Consider the unlabeled pair $p_8=(o_5, o_6)$. In the \clustergraph, since $o_5$ and $o_6$ are in different clusters, and $\cluster{o_5}$ and $\cluster{o_6}$ have an edge, then there must exist a path from $o_5$ to $o_6$ which contains a single non-matching pair (e.g., $o_5\!\rightarrow\!o_4\!\rightarrow\!o_6$), thus $p_8=(o_5, o_6)$ can be deduced as a non-matching pair from $\{p_1, p_2, \cdots, p_7\}$.
\end{example}


The labeling algorithm needs to enumerate each pair one by one and it can only publish a single pair to the crowdsourcing platform. Hence, every time there is only one available HIT in the crowdsourcing platform. This constraint makes the workers unable to do HIT simultaneously and results in long latency. Notice that batching strategies~\cite{journals/pvldb/MarcusWKMM11,journals/pvldb/WangKFF12}, which place multiple pairs into a single HIT, have been proved useful in reducing the money cost. However, the simple approach is unable to support the batching techniques since only one pair is allowed to publish every time, and thus results in more money cost. To overcome these drawbacks, we propose a parallel labeling algorithm in Section~\ref{sec:parallel}, which can crowdsource multiple pairs every time without increasing the total number of crowdsourced pairs.

\section{Sorting}\label{sec:sorting}

As observed in Section~\ref{subsec:sorting}, different labeling orders result in different numbers of crowdsourced pairs. In this section, we explore the optimal labeling order to minimize the number of crowdsourced pairs. We begin with the formulation of this problem and discuss how to find the optimal labeling order in Section~\ref{subsec:optimal-order}. However, the optimal order cannot be achieved in reality. Thus, in Section~\ref{subsec:expect-optimal-order}, we propose a heuristic labeling order. 

\subsection{Optimal Labeling Order}\label{subsec:optimal-order}

Given a labeling order $\order=\langle p_1, p_2, \cdots, p_n \rangle$, let $\cn(\order)$ denote the number of crowdsourced pairs required by $\order$. Our goal is to identify the optimal labeling order which results in the minimum number of crowdsourced pairs. The formal definition of this problem is as follows.

\begin{definition}[Optimal Labeling Order]
Given a set of object pairs, the problem of identifying the optimal labeling order is to get a sorted list of pairs, $\order_{op}$, such that the number of crowdsourced pairs is minimal using the order, i.e. $\cn(\order_{op}) \leq \cn(\order)$ holds for any other order $\order$.
\end{definition}

For example, assume three object pairs, $p_1=(o_1, o_2)$, $p_2=(o_2, o_3)$, $p_3=(o_1, o_3)$, where $p_1$ is a matching pair, and $p_2, p_3$ are two non-matching pairs. We can label them in six different labeling orders, $\order_1= \langle p_1, p_2, p_3\rangle$, $\order_2= \langle p_1, p_3, p_2\rangle$, $\order_3=\langle p_2, p_3, p_1\rangle$, $\order_4= \langle p_2, p_1, p_3\rangle$, $\order_5= \langle p_3, p_1, p_2\rangle$, and $\order_6=\langle p_3, p_2, p_1\rangle$. The numbers of crowdsourced pairs of the six orders are respectively $\cn(\order_1)=2$, $\cn(\order_2)=2$, $\cn(\order_3)=3$, $\cn(\order_4)=2$, $\cn(\order_5)=2$, and $\cn(\order_6)=3$. As $\order_1$, $\order_2$, $\order_4$, $\order_5$ lead to the minimum number of crowdsourced pairs, any one of them can be seen as the optimal labeling order.

Notice that any labeling order can be changed to another labeling order by swapping adjacent pairs. We first study how swapping two adjacent pairs affects the number of crowdsourced pairs. We have an observation that it is always \emph{better} to first label a matching pair and then a non-matching pair, i.e., it will lead to fewer or equal number of crowdsourced pairs. Recall $\order_3=\langle p_2, p_3, p_1\rangle$ and $\order_4=\langle p_2, p_1, p_3\rangle$ in the above example. $\order_3$ labels a non-matching pair $p_3$ before a matching pair $p_1$ while $\order_4$ swaps the positions of $p_1$ and $p_3$, and labels a matching pair $p_1$ before a non-matching pair $p_3$. As $\cn(\order_3)=3$ and $\cn(\order_4)=2$, the example shows that $\order_4$, which first labels a matching pair, needs to crowdsource fewer pairs than $\order_3$. Lemma~\ref{lem:matching-better} formulates this idea.


\begin{lemma}\label{lem:matching-better}
Consider two labeling orders,
\[\order=\langle p_1, \cdots, p_{i-1},  \bm{p_i}, \bm{p_{i+1}}, p_{i+2}, \cdots p_n \rangle, \]
\[\order'=\langle p_1, \cdots, p_{i-1}, \bm{p_{i+1}}, \bm{p_{i}}, p_{i+2}, \cdots p_n \rangle,\]
where $\order'$ is obtained by swapping $p_{i}$ and $p_{i+1}$ in $\order$. If $p_{i}$ is a non-matching pair, and $p_{i+1}$ is a matching pair, then we have $\cn(\order')\leq \cn(\order)$.
\end{lemma}
\begin{proof}
For a pair $p_j$, if $j\not\in\{i, i+1\}$, it is easy to see that $\{p_1$, $p_2$, $\cdots$, $p_{j-1}\}$ in $\order$ is the same as that in $\order'$. Therefore, if $p_j$ is a crowdsourced pair in $\order$, it must be a crowdsourced pair in $\order'$, and vice versa. Hence, we only need to check whether $p_i, p_{i+1}$ are crowdsourced pairs in $\order$ and $\order'$.

There are four possible cases for $p_i,p_{i+1}$ in $\order$: deduced and crowdsourced pairs, crowdsourced and crowdsourced pairs, deduced and deduced pairs, or crowdsourced and deduced pairs. We prove that in any case, $\{p_i,p_{i+1}\}$ in $\order'$ would contain fewer or equal number of crowdsourced pairs than $\{p_i,p_{i+1}\}$ in $\order$, thus $\cn(\order')\leq \cn(\order)$ holds.

\vspace{.5em}

{\noindent {\bf Case 1:} \emph{$p_i$ is deduced and $p_{i+1}$ is crowdsourced.}} As $p_i$ in $\order$ can be deduced from $\{p_1, p_2, \cdots,p_{i-1}\}$, in the graph built for them, there is a path from one object of $p_i$ to the other with no more than one non-matching pairs. Such path still exists in the graph built for more pairs $\{p_1, p_2, \cdots,p_{i-1}, p_{i+1}\}$, thus $p_{i}$ in $\order'$ is also a deduced pair. Similarly, as $p_{i+1}$ in $\order$ cannot be deduced from $\{p_1, p_2, \cdots,p_{i}\}$, it cannot be deduced from fewer pairs $\{p_1, p_2, \cdots,p_{i-1}\}$, thus $p_{i+1}$ in $\order'$ is also a crowdsourced pair. Therefore, $\cn(\order')\leq \cn(\order)$ holds in Case~1.

\vspace{.25em}

{\noindent {\bf Case 2:} \emph{$p_i$ is crowdsourced and $p_{i+1}$ is crowdsourced.}} In the worst case, both $p_i$ and $p_{i+1}$ in $\order'$ still need to be crowdsourced, so $\cn(\order')\leq \cn(\order)$ holds in Case~2.

\vspace{.25em}

\sloppy

{\noindent {\bf Case 3:} \emph{$p_i$ is deduced and $p_{i+1}$ is deduced.}} As $p_i$ in $\order$ can be deduced from $\{p_1, p_2, \cdots,p_{i-1}\}$, it can also be deduced from more pairs $\{p_1, p_2, \cdots,p_{i-1}, p_{i+1}\}$, thus $p_{i}$ in $\order'$ is also a deduced pair. As $p_{i+1}$ in $\order$ can be deduced from $\{p_1, p_2, \cdots, p_i\}$, it is easy to derive that $p_{i+1}$ can also be deduced from the \emph{crowdsourced} pairs in $\{p_1, p_2, \cdots, p_i\}$. Since $p_i$ is not a \emph{crowdsourced} pair based on the given condition, $p_{i+1}$ can be deduced from the \emph{crowdsourced} pairs in $\{p_1, p_2, \cdots, p_{i-1}\}$, thus $p_{i+1}$ in $\order'$ is a deduced pair. As $p_i$ and $p_{i+1}$ in $\order'$ are both deduced pairs, $\cn(\order')\leq \cn(\order)$ holds in Case~3.

\fussy

\vspace{.25em}

{\noindent {\bf Case 4:} \emph{$p_i$ is crowdsourced and $p_{i+1}$ is deduced.}} As the matching pair $p_{i+1}$ in $\order$ can be deduced from $\{p_1, p_2, \cdots,p_{i}\}$, in the graph built for $\{p_1, p_2, \cdots,p_{i}\}$, there exists a path from one object of $p_{i+1}$ to the other which only consists of matching pairs. As $p_{i}$ is a non-matching pair, after removing $p_{i}$ from the graph, the path still exists in the graph built for $\{p_1, p_2, \cdots,p_{i-1}\}$, thus $p_{i+1}$ can be deduced from $\{p_1, p_2, \cdots,p_{i-1}\}$. Hence, $p_{i+1}$ in $\order'$ is also a deduced pair. In the worst case, $p_{i}$ in $\order'$ still need to be crowdsourced, so $\cn(\order')\leq \cn(\order)$ holds in Case~4.
\end{proof}

For a given labeling order, by using the above method to swap adjacent object pairs, we can put all the matching pairs prior to the non-matching pairs. As it is better to label a matching pair first and then a non-matching pair, the new order will require fewer or equal number of crowdsourced pairs than the original labeling order.

Next we prove that swapping adjacent matching pairs or non-matching pairs will not change the number of crowdsourced pairs. 


\begin{lemma}\label{lem:matching-nochange}
Consider two labeling orders,
\[\order=\langle p_1, \cdots, p_{i-1},  \bm{p_i}, \bm{p_{i+1}}, p_{i+2}, \cdots p_n \rangle, \]
\[\order'=\langle p_1, \cdots, p_{i-1}, \bm{p_{i+1}}, \bm{p_{i}}, p_{i+2}, \cdots p_n \rangle,\]
where $\order'$ is obtained by swapping $p_{i}$ and $p_{i+1}$ in $\order$. If $p_{i}$ and $p_{i+1}$ are both matching pairs or both non-matching pairs, then we have $\cn(\order') = \cn(\order)$.
\end{lemma}
\begin{proof}
Since swapping $p_i$ and $p_{i+1}$ will not affect the other pairs except $p_i$ and $p_{i+1}$ (see the proof in Lemma~\ref{lem:matching-better}), we only need to prove that $\{p_i, p_{i+1}\}$ in $\order'$ requires the same number of crowdsourced pairs as $\{p_i, p_{i+1}\}$ in $\order$, thus $\cn(\order')=\cn(\order)$ holds. We still consider the four cases.

\vspace{.5em}

{\noindent {\bf Case 1:} \emph{$p_i$ is deduced and $p_{i+1}$ is crowdsourced.}} As in the proof of Case 1 in Lemma~\ref{lem:matching-better}, we have $p_{i}$ in $\order'$ is also a deduced pair, and $p_{i+1}$ in $\order'$ is also a crowdsourced pair. Therefore, $\cn(\order')= \cn(\order)$ holds in Case~1.

\vspace{.25em}

\sloppy

{\noindent {\bf Case 2:} \emph{$p_i$ is crowdsourced and $p_{i+1}$ is crowdsourced.}} As $p_{i+1}$ in $\order$ cannot be deduced from $\{p_1, p_2, \cdots,p_{i}\}$, it cannot be deduced from fewer pairs $\{p_1, p_2, \cdots,p_{i-1}\}$, thus $p_{i+1}$ in $\order'$ is also a crowdsourced pair.

Next, we prove by contradiction $p_i$ in $\order'$ is also a crowdsourced pair. Assume $p_i$ in $\order'$ is not a crowdsourced pair. Then $p_i$ can be deduced from $\{p_1, p_2, \cdots, p_{i-1}, p_{i+1}\}$. In the graph built for $\{p_1, p_2, \cdots, p_{i-1}, p_{i+1}\}$, there exists a path from one object of $p_{i}$ to the other which contains no more than one non-matching pair. And since $p_i$ in $\order$ is crowdsourced, $p_i$ cannot be deduced from $\{p_1, p_2, \cdots, p_{i-1}\}$, thus $p_{i+1}$ must be in the path. By removing $p_{i+1}$ from the path and adding $p_{i}$ to the path, we obtain a new path from one object of $p_{i+1}$ to the other. Since $p_{i+1}$ is removed and $p_{i}$ is added, the path must be in the graph built for $\{p_1, p_2, \cdots, p_{i}\}$. As $p_{i}$, $p_{i+1}$ are either both matching pairs or both non-matching pairs, the path contains no more than one non-matching pair. Therefore, $p_{i+1}$ can be deduced from $\{p_1, p_2, \cdots, p_{i}\}$ which contradicts $p_{i+1}$ in $\order$ is a crowdsourced pair. Hence, the assumption does not hold, and $p_i$ in $\order'$ is a crowdsourced pair. Since $p{_i},p_{i+1}$ in $\order'$ are both crowdsourced pairs, $\cn(\order')= \cn(\order)$ holds in Case~2.

\fussy

\vspace{.25em}

{\noindent {\bf Case 3:} \emph{$p_i$ is deduced and $p_{i+1}$ is deduced.}} As in the proof of Case 3 in Lemma~\ref{lem:matching-better}, we have $p_{i}$ in $\order'$ is also a deduced pair, and $p_{i+1}$ in $\order'$ is also a deduced pair. Therefore, $\cn(\order') = \cn(\order)$ holds in Case~3.

\vspace{.25em}

{\noindent {\bf Case 4:} \emph{$p_i$ is crowdsourced and $p_{i+1}$ is deduced.}} As $p_{i+1}$ in $\order$ can be deduced from $\{ p_1, p_2, \cdots, p_i\}$, in the graph built for $\{p_1, p_2, \cdots, p_{i}\}$, there exist some paths from one object of $p_{i+1}$ to the other which contains no more than one non-matching pair. There are three cases about these paths:

(a) If none of these paths contains $p_i$, that is, $p_i$ and $p_{i+1}$ will not affect each other, then $p_i$ in $\order'$ is also a crowdsourced pair and $p_{i+1}$ in $\order'$ is also a deduced pair. Therefore, $\cn(\order')= \cn(\order)$ holds in Case~4(a).

(b) If some of these paths contain $p_i$ but others do not, then we can infer that in the graph built for $\{p_1, p_2, \cdots, p_{i-1}\}$, there exists a path from one object of $p_i$ to the other with no more than one non-matching pair, thus $p_i$ can be deduced from $\{p_1, p_2, \cdots, p_{i-1}\}$ which contradicts $p_i$ in $\order$ is a crowdsourced pair. Hence, Case~4(b) is impossible.

(c) If all of the paths contain $p_i$, after removing $p_i$ from these paths, in the graph built for $\{p_1, p_2, \cdots, p_{i-1}\}$, there will be no path from one object of $p_{i+1}$ to the other which contains no more than one non-matching pair, thus $p_{i+1}$ cannot be deduced from $\{p_1, p_2, \cdots, p_{i-1}\}$. That is, $p_{i+1}$ in $\order'$ is a crowdsourced pair. Next we prove that $p_i$ in $\order'$ is a deduced pair. Consider one of these paths that contain $p_i$. After removing $p_i$ from the path and adding $p_{i+1}$ to the path, we obtain a new path from one object of $p_{i}$ to the other. Since $p_{i}$ is removed and $p_{i+1}$ is added, the path must be in the graph built for $\{p_1, p_2, \cdots, p_{i-1}, p_{i+1}\}$. As $p_{i}$, $p_{i+1}$ are either both matching pairs or both non-matching pairs, the path contains no more than one non-matching pair. Therefore, $p_{i}$ can be deduced from $\{p_1, p_2, \cdots, p_{i-1}, p_{i+1}\}$, thus it is a deduced pair in $\order'$. As $p_i$ in $\order'$ is a deduced pair, and $p_{i+1}$ in $\order'$ is a crowdsourced pair, $\cn(\order')= \cn(\order)$ holds in Case~4(c).
\end{proof}

For two different labeling orders, if they both first label all the matching pairs and then label the other non-matching pairs, we can change one labeling order to the other by swapping adjacent matching pairs and adjacent non-matching pairs. Based on Lemma~\ref{lem:matching-nochange}, the two labeling orders require the same number of crowdsourced pairs. Therefore, any labeling order, which puts all the matching pairs to the front of the other non-matching pairs, is the optimal.


\begin{theorem}\label{thm:optimal-order}
Given a set of object pairs, the optimal labeling order is to first label all the matching pairs, and then label the other non-matching pairs.
\end{theorem}

For example, consider the pairs in Figure~\ref{fig:running-example}. $\order_{op}$ = $\langle p_1$, $p_2$, $p_4$, $p_5$, $p_3$, $p_6$, $p_7$, $p_8 \rangle$ is the optimal labeling order since all the matching pairs, i.e., $p_1, p_2, p_4, p_5$, are labeled before the other non-matching pairs, i.e., $p_3, p_6, p_7, p_8$.

Now we have proved that the optimal labeling order is to first label all the matching pairs, and then label the other non-matching pairs. However, when identifying the labeling order, we have no idea about whether a pair is matching or non-matching, therefore, the optimal labeling order cannot be achieved in reality. To address this problem, we investigate an expected optimal labeling order in the next section.

\subsection{Expected Optimal Labeling Order}\label{subsec:expect-optimal-order}

In this section, we aim to identify a labeling order that requires as few crowdsourced pairs as possible. Recall the optimal labeling order which first labels the matching pairs and then labels the non-matching pairs. Although we do not know the real matching pairs upfront, machine-based methods can be applied to compute for each pair the likelihood that they are matching. For example, the likelihood can be the similarity computed by a given similarity function~\cite{journals/pvldb/WangKFF12}.

Consider a labeling order $\order = \langle p_1, p_2, \cdots, p_n \rangle$. Suppose each pair in $\order$ is assigned with a probability that they are matching. Then the number of crowdsourced pairs required by $\order$ becomes a random variable. Its expected value is computed as the sum of the probability that $p_i$ is a crowdsourced pair $(1 \leq i \leq n)$, i.e.,
\[ \ecn\big[\cn(\order)\big] = \sum_{i=1}^{n} \mathbb{P} (p_i = \textsf{crowdsourced}). \]
To compute $\mathbb{P}(p_i = \textsf{crowdsourced})$, we enumerate the possible labels of $\{p_1, p_2, \cdots, p_{n}\}$, and for each possibility, since the labels of $\{p_1, p_2, \cdots, p_{i-1}\}$ are known, we can derive whether $p_i$ is a crowdsourced pair or not. Hence, $\mathbb{P}(p_i = \textsf{crowdsourced})$ is the sum of the probability of each possibility that $p_i$ is a crowdsourced pair.


We aim to identify a labeling order that can minimize the expected number of crowdsourced pairs since the order is expected to require the minimum number of crowdsourced pairs. We call such an order an \emph{expected optimal labeling order}. The following definition formulates this problem.

\sloppy

\begin{definition}[Expected Optimal Labeling Order]
Given a set of object pairs, and each object pair is assigned with a probability that they are matching, the problem of identifying the expected optimal labeling order $\order_{eop}$ is to compute a sorted list of pairs such that the expected number of crowdsourced pairs is minimal using the order, i.e. $\ecn\big[\cn(\order_{eop})\big]\leq \ecn\big[\cn(\order)\big]$ holds for any other order $\order$.
\end{definition}

\begin{example}\label{exa:expect}
Consider three pairs, $p_1=(o_1, o_2)$, $p_2=(o_2, o_3)$ and $p_3=(o_1, o_3)$. Suppose the probabilities that $p_1$, $p_2$ and $p_3$ are matching pairs are respectively 0.9, 0.5 and 0.1. There are six different labeling orders, $\order_1= \langle p_1, p_2, p_3\rangle$, $\order_2= \langle p_1, p_3, p_2\rangle$, $\order_3= \langle p_2, p_3, p_1\rangle$, $\order_4=\langle p_2, p_1, p_3\rangle$, $\order_5= \langle p_3, p_1, p_2\rangle$, and $\order_6=\langle p_3, p_2, p_1\rangle$. We first compute the excepted number of crowdsourced pairs for $\order_1= \langle p_1, p_2, p_3\rangle$. For the first pair $p_1$, as there are no labeled pairs, it must need crowdsourcing, thus $\mathbb{P} (p_1\!=\! \textsf{crowdsourced})=1$. For the second pair $p_2$, as $p_2 = (o_2, o_3)$ cannot be deduced from $p_1 = (o_1, o_2)$, it must need crowdsourcing, thus $\mathbb{P} (p_2 = \textsf{crowdsourced}) = 1$. 
For the third pair $p_3$, we enumerate the possible labels of $\{p_1, p_2, p_3\}$, i.e., \{\textrm{matching, matching, matching}\},  \{\textrm{non-matching, matching, non-matching}\}, \{\textrm{matching, non-matching, non-matching}\}, \{\textrm{non-matching, non-matching, matching}\}, \{\textrm{non-matching, non-matching, non-matching}\}. Among the five possibilities, $p_3$ needs to be crowdsourced only when both $p_1$ and $p_2$ are non-matching pairs (i.e., the last two possibilities). Hence, the probability that $p_3$ is a crowdsourced pair is $\frac{0.1*0.5*0.1+0.1*0.5*0.9}{0.9*0.5*0.1+0.1*0.5*0.9+0.9*0.5*0.9+0.1*0.5*0.1+0.1*0.5*0.9} = 0.09$.
By summing up the probabilities that $p_1, p_2, p_3$ are crowdsourced pairs, we have $\ecn\big[\cn(\order_1)\big] = 1+1+0.09 = 2.09$. Similarly, we can compute $\ecn\big[\cn(\order_2)\big]=2.17$, $\ecn\big[\cn(\order_3)\big]=2.83$, $\ecn\big[\cn(\order_4)\big]=2.09$, $\ecn\big[\cn(\order_5)\big]=2.17$, and $\ecn\big[\cn(\order_6)\big]=2.83$. As $\order_1$ and $\order_4$ require the minimum expected number of crowdsourced pairs, either one of them can be taken as the expected optimal labeling order.
\end{example}

A recent VLDB paper has proved that the problem of identifying the expected optimal labeling order is NP-hard~\cite{DBLP:journals/pvldb/NorasesBD14}. In our paper, we propose a heuristic method to solve this problem. Recall the analysis of Section~\ref{subsec:optimal-order}, we have proved that it is better to label a matching pair before a non-matching pair (Lemma~\ref{lem:matching-better}). This idea inspires us to label the object pairs in the decreasing order of the likelihood that they are matching. For example, consider the unlabeled pairs $p_1, p_2, \cdots, p_8$ in Figure~\ref{fig:running-example}. To identify their labeling order, we first use a machine-based method to compute a likelihood for each pair that it is a matching pair, and then label the pairs in the decreasing order of the likelihood, i.e., $\order_{eop}$ = $\langle p_1$, $p_2$, $p_3$, $p_4$, $p_5$, $p_6$, $p_7$, $p_8 \rangle$.

\fussy


\sloppy

\section{Parallel Labeling} \label{sec:parallel}

After identifying a labeling order, our labeling framework will label the unlabeled pairs in this order. In Section~\ref{subsec:labeling}, we present a simple approach to achieve this goal. However, the approach only allows to publish a single pair to the crowdsourcing platform, which is unable to label the pairs simultaneously and results in long latency. To alleviate this problem, we propose a parallel labeling algorithm in Section~\ref{subsec:parallel-label-crowdsource}, which can crowdsource multiple pairs every time without increasing the total number of required crowdsourced pairs. To further improve the parallelism, we present two optimization techniques in Section~\ref{subsec:parallel-optimization}.


\subsection{Parallel Labeling Algorithm}  \label{subsec:parallel-label-crowdsource}

We first use an example to show our basic idea. Consider the labeling order $\order = \big\langle (o_1, o_2), (o_2, o_3), (o_3, o_4)\big\rangle$. The simple labeling approach will first crowdsource the first pair $(o_1, o_2)$, and cannot crowdsource the second pair until the first pair is labeled. However, for the second pair $(o_2, o_3)$, we observe that no matter which label the first pair gets, we must need to crowdsource it since the second pair $(o_2, o_3)$ cannot be deduced from the first pair $(o_1, o_2)$. For the third pair $(o_3, o_4)$, we have a similar observation that no matter which labels the first two pairs get, we must crowdsource it since the third pair $(o_3, o_4)$ cannot be deduced from the first two pairs $(o_1, o_2)$ and $(o_2, o_3)$. Therefore, all the pairs in $\order$ can be crowdsourced together instead of individually. Based on this idea, we propose a parallel labeling algorithm as shown in Figure~\ref{algo:parallel-labeling}.

\vspace{.5em}

{\noindent {\bf Algorithm Overview:}} Our parallel labeling algorithm employs an iterative strategy. In each iteration, the algorithm first identifies a set of pairs that can be crowdsourced in parallel~(Line~\ref{algo:parallel-labeling:parallel}). Then the algorithm publishes the pairs simultaneously to the crowdsourcing platform, and obtains their crowdsourced labels~(Line~\ref{algo:parallel-labeling:crowdsource-label}). After that, the algorithm utilizes the already labeled pairs to deduce subsequent unlabeled pairs~(Lines~\ref{algo:parallel-labeling:deduce-label-begin}-\ref{algo:parallel-labeling:deduce-label-end}). The algorithm repeats the iterative process until all the pairs are labeled.

\begin{figure}[tup]
\begin{algorithm}[H]

\linesnumbered \SetVline

\small
\caption{\textsc{ParallelLabeling}($\order$)}

\SetLine

\KwIn{$\order = \langle p_1, p_2, \cdots, p_n \rangle$~: a sorted list of unlabeled pairs}
 \KwOut{$\labelset = \{(p_i, \ell)~|~1 \leq i\leq n\}$: a set of labeled pairs}

\SetVline
\Begin{
    $\labelset = \{\}$\;
    \While {\emph{there is an unlabeled pair in $\order$}}{
        $\publish$ = \textsf{ParallelCrowdsourcedPairs}($\order$);\nllabel{algo:parallel-labeling:parallel}\\
        $\labelset$ $\cup$= \textsf{CrowdsourceLabels}($\publish$);\nllabel{algo:parallel-labeling:crowdsource-label}\\
        \For {\emph{each unlabeled pair $p\in\order$}}{\nllabel{algo:parallel-labeling:deduce-label-begin}
            \If {\emph{\textsf{DeducedLabel}($p$, $\labelset$)}}{
                Add $(p, \ell)$ into $\labelset$;\nllabel{algo:parallel-labeling:deduce-label-end}\\
            }
        }
    }
    \textbf{return} $\labelset$\;

}
\end{algorithm} \vspace{-1.5em}
\caption{Parallel labeling algorithm.}\label{algo:parallel-labeling} \vspace{-1em}
\end{figure}

\vspace{.5em}

A big challenge in the algorithm is to identify a set of pairs that can be crowdsourced in parallel. We know that a pair $p_i = (o, o')$ needs to be crowdsourced if and only if $p_i$ cannot be deduced from $\{p_1, p_2, \cdots, p_{i-1}\}$.

In the case that $\{p_1, p_2, \cdots, p_{i-1}\}$ are labeled, we need to check the graph built for $\{ p_1, p_2, \cdots, p_{i-1}\}$. If every path from $o$ to $o'$ contains more than one non-matching pair, then $p_i$ cannot be deduced from $\{p_1, p_2, \cdots, p_{i-1}\}$. However, some pairs in $\{p_1, p_2, \cdots, p_{i-1}\}$ may have not been labeled. In this case, we have no idea about the exact number of non-matching pairs in some paths. In order to see if every path must contain more than one non-matching pair, we compute the minimum number of non-matching pairs in each path by supposing all the unlabeled pairs are matching pairs. If the minimum number of non-matching pairs in each path is larger than one, then whatever the unlabeled pairs in $\{p_1, p_2, \cdots, p_{i-1}\}$ are labeled, the number of non-matching pairs in each path must be larger than one, thus $p_i$ cannot be deduced from $\{p_1, p_2, \cdots, p_{i-1}\}$, and needs to be crowdsourced.

Based on this idea, given a sorted list of object pairs, $\order = \langle p_1, p_2, \cdots, p_n\rangle$, where some pairs have not been labeled, to identify which pairs can be crowdsourced in parallel, we first suppose all the unlabeled pairs are matching pairs, and then for each pair $p_i$ $(1\leq i \leq n)$, we output $p_i$ as a crowdsourced pair if it cannot be deduced from $\{p_1, p_2, \cdots, p_{i-1}\}$. As discussed in Section~\ref{subsec:labeling}, \clustergraph can be utilized to efficiently decide whether $p_i$ cannot be deduced from $\{p_1, p_2, \cdots, p_{i-1}\}$. Note that to make a decision for each $p_i$ ($1\leq i \leq n$), we do not need to build the \clustergraph from scratch since the \clustergraph can be easily obtained by inserting $p_{i-1}$ into the \clustergraph built for $\{p_1, p_2, \cdots, p_{i-2}\}$.


Figure~\ref{algo:parallel-crowdsourcing-pair} shows the algorithm for identifying the pairs that can be crowdsourced in parallel in each iteration. The algorithm first initializes an empty \clustergraph, and then checks each pair $p_i$ ($1\leq i \leq n$). If $p_i$ has already been labeled, it does not need to be crowdsourced any more, thus we update the \clustergraph by inserting $p_i$, and go to the next pair $p_{i+1}$ (Lines~\ref{algo:parallel-crowdsourcing-pair:labeled-begin}-\ref{algo:parallel-crowdsourcing-pair:labeled-end}); otherwise, we check whether $p_i$ can be deduced from $\{p_1, p_2, \cdots, p_{i-1}\}$. In the \clustergraph, if $\cluster{o} \neq \cluster{o'}$ and there is no edge between $\cluster{o}$ and $\cluster{o'}$, $p_i$ cannot be deduced, and thus needs to be crowdsourced (Lines~\ref{algo:parallel-crowdsourcing-pair:crowdsource-begin}-\ref{algo:parallel-crowdsourcing-pair:crowdsource-end}). In this case, since $p_i$ is unlabeled, we suppose it is a matching pair, and insert $p_i$ into the \clustergraph, and go to the next pair $p_{i+1}$ (Line~\ref{algo:parallel-crowdsourcing-pair:insert}). After checking all the pairs, we return the obtained crowdsourced pairs.

\begin{figure}[tup] \vspace{.5em}
\decmargin{.5em}
\begin{algorithm}[H]

\linesnumbered \SetVline

\small

\caption{\textsc{ParallelCrowdsourcedPairs}($\order$, $\labelset$)}

\SetLine

\KwIn{$\order = \langle p_1, p_2, \cdots, p_n \rangle$~: a sorted list of unlabeled pairs\\
      \hspace{3.3em}$\labelset$: a set of labeled pairs}
 \KwOut{$\publish$: a set of pairs that can be crowdsourced in~parallel}

\SetVline

\Begin{
    $\publish = \{\}$;\nllabel{algo:parallel-labeling:publish-begin}\\
    Initialize an empty \clustergraph\;
    \For {\emph{$i=1$ to $n$}}{
        \If {\emph{$(p_i, \ell) \in \labelset$}}{\nllabel{algo:parallel-crowdsourcing-pair:labeled-begin}
            Insert $(p_i, \ell)$ into \clustergraph;\nllabel{algo:parallel-crowdsourcing-pair:labeled-end}\\
        }
        \Else{

            $p_i = (o, o')$\;
            \If {\emph{$\cluster{o} \neq \cluster{o'}$ \textbf{and} there is no edge between $\cluster{o}$ and $\cluster{o'}$}}{\nllabel{algo:parallel-crowdsourcing-pair:crowdsource-begin}
                Add $p_i$ into $\publish$;\nllabel{algo:parallel-crowdsourcing-pair:crowdsource-end}\\
             }
            Insert $(p_i, ``\textrm{matching}")$ into \clustergraph;\nllabel{algo:parallel-crowdsourcing-pair:insert}\\
        }
    }
    \textbf{return} $\publish$;\nllabel{algo:parallel-crowdsourcing-pair:return}\\

}
\end{algorithm} \vspace{-1em}
\incmargin{.5em}
\caption{\textsf{ParallelCrowdsourcedPairs} algorithm.}\label{algo:parallel-crowdsourcing-pair} \vspace{-1em}
\end{figure}



\begin{figure}[htbp]
\centering
  \includegraphics[scale=0.25]{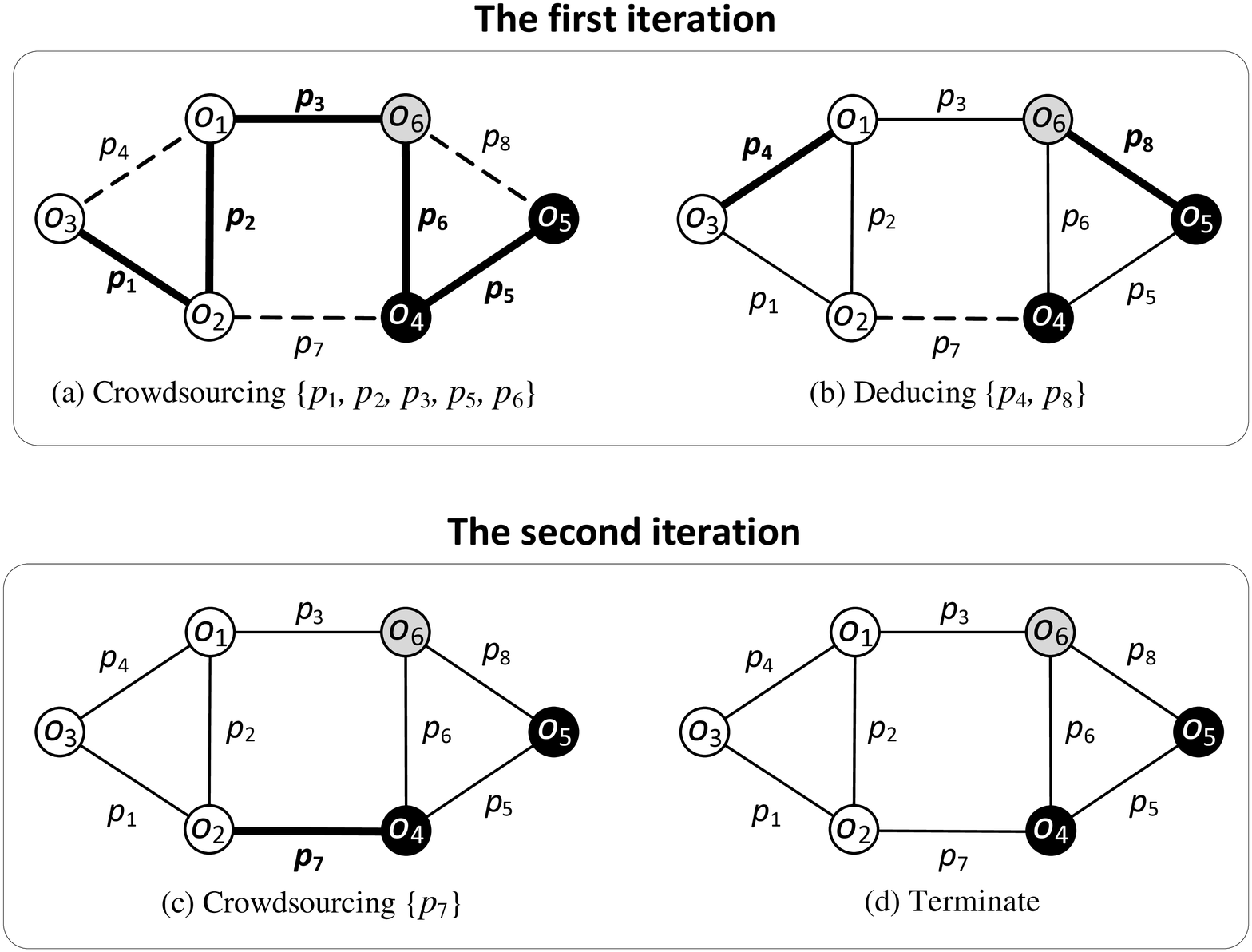}\\\vspace{-.5em}
  \caption{An illustration of parallel labeling algorithm.}\label{fig:parallel-labeling}\vspace{-1em}
\end{figure}

\begin{example}
Consider the running example in Figure~\ref{fig:running-example}. Given the labeling order $\langle p_1, p_2, \cdots, p_8 \rangle$, Figure~\ref{fig:parallel-labeling} shows how to use the parallel labeling algorithm to label the pairs (The solid edges represent labeled pairs and the dotted edges represent unlabeled pairs.).

In the first iteration, since all the given pairs are unlabeled, we first suppose $\{p_1, p_2, \cdots, p_8\}$ are matching pairs, and then identify five pairs (i.e., $p_1$, $p_2$, $p_3$, $p_5$, $p_6)$ that can be crowdsourced in parallel (the bold solid edges in Figure~\ref{fig:parallel-labeling}(a)). For example, $p_5$ is identified since it cannot be deduced from $\{p_1, p_2, p_3, p_4\}$ while $p_7$ is not identified since it can be deduced from $\{ p_1, p_2, p_3, p_4, p_5, p_6 \}$ (Note that they are supposed as matching pairs). We are able to publish $\{p_1$, $p_2$, $p_3$, $p_5$, $p_6\}$ simultaneously to the crowdsourcing platform. After obtaining their labels, based on transitive relations, we can deduce $p_4$ from $p_1$ and $p_2$, and deduce $p_8$ from $p_5$ and $p_6$ (the bold solid edges in Figure~\ref{fig:parallel-labeling}(b)). Since there still exists an unlabeled pair (i.e. $p_7$), we repeat the iteration process.

In the second iteration, since only $p_7$ is unlabeled, we first suppose $p_7$ is a matching pair, and then identify one pair (i.e., $p_7$) for crowdsourcing (the bold-solid edges in Figure~\ref{fig:parallel-labeling}(c)). We publish $p_7$ to the crowdsourcing platform. After it is labeled, we find all the pairs have been labeled, thus the algorithm is terminated, and the labeled pairs are returned.
\end{example}

\vspace{-1em}

\subsection{Optimization Techniques}  \label{subsec:parallel-optimization}

In this section, we propose two optimization techniques, instant decision and non-matching first, to further enhance our parallel labeling algorithm.

\vspace{.5em}

{\noindent \bf Instant Decision:} Recall our parallel labeling algorithm. The algorithm will first publish some pairs to the crowdsourcing platform, and after all the published pairs have been labeled, decide which pairs can be crowdsourced next. Notice that we do not need to wait until all the published pairs have been labeled to decide the next-round crowdsourced pairs. Instead when some of the published pairs are labeled, we can utilize them instantly to crowdsource the remaining pairs. For example, in Figure~\ref{fig:parallel-labeling}, we first publish $\{p_1$, $p_2$, $p_3$, $p_5$, $p_6\}$ together to the crowdsourcing platform. If $p_3$ and $p_6$ are labeled, we can deduce that $p_7 = (o_2, o_4)$ must be a crowdsourced pair, and can be published instantly instead of waiting for the other pairs. This is because, in the graph built for $\{p_1, p_2, p_3, p_4, p_5, p_6\}$, there are two paths from $o_2$ to $o_4$, i.e., $o_2\!\rightarrow\!o_1\!\rightarrow\!o_6\!\rightarrow\!o_4$ and $o_2\!\rightarrow\!o_3\!\rightarrow\!o_1\!\rightarrow\!o_6\!\rightarrow\!o_4$. Both paths contain at least two non-matching pairs (i.e., $p_3$ and $p_6$), thus $p_7$ cannot be deduced from $\{p_1$, $p_2$, $p_3$, $p_4$, $p_5$, $p_6\}$ based on transitive relations.

Based on this idea, we propose an optimization technique, called \emph{instant decision}, which will make an instant decision on which pairs can be published next whenever a single published pair (instead of all the published pairs) is labeled. Achieving this goal requires a minor change to the algorithm in Figure~\ref{algo:parallel-crowdsourcing-pair} by excluding the already published pairs from $\publish$ in Line~\ref{algo:parallel-crowdsourcing-pair:return}. By applying the optimization technique to our parallel labeling algorithm, we are able to increase the number of the available pairs in the crowdsourcing platform to enhance the effect of parallelism.

\vspace{.5em}

{\noindent \bf Non-matching First:} If we utilize the instant-decision optimization technique, when a published pair is labeled, we need to decide which pairs can be crowdsourced next. We find if the labeled pair is a matching pair, that will not lead to publishing any other pair. This is because when deciding which pairs can be crowdsourced in Figure~\ref{algo:parallel-crowdsourcing-pair}, we have assumed that all the unlabeled pairs are matching pairs. Hence, knowing an unlabeled pair is a matching pair will have no effect on the algorithm. Based on this idea, we propose an optimization technique, called \emph{non-matching first}. Consider the published pairs in the crowdsourcing platform. If we could ask the crowd workers to label the potentially non-matching pairs first, i.e., label the published pairs in the increasing order of the probability that they are a matching pair, that would increase the number of the available pairs in the crowdsourcing platform so as to enhance the effect of parallelism. It is worth noting that this order is for the published pairs in the parallel labeling algorithm, which is different from the order for labeling all pairs in Section~\ref{sec:sorting}. 



\vspace{-.5em}
\section{Experiment}
\label{sec:exp}
In this section, we evaluate our method. The goals of the experiments are to (1) examine the effectiveness of transitive relations in reducing the number of crowdsourced pairs, (2) compare the number of crowdsourced pairs required by different labeling orders, (3) validate the advantage of our parallel labeling algorithm over the non-parallel labeling algorithm, and (4) illustrate the performance of our method in a real crowdsourcing platform.

\begin{figure}[tbp] 
\centering
  \includegraphics[scale=0.7]{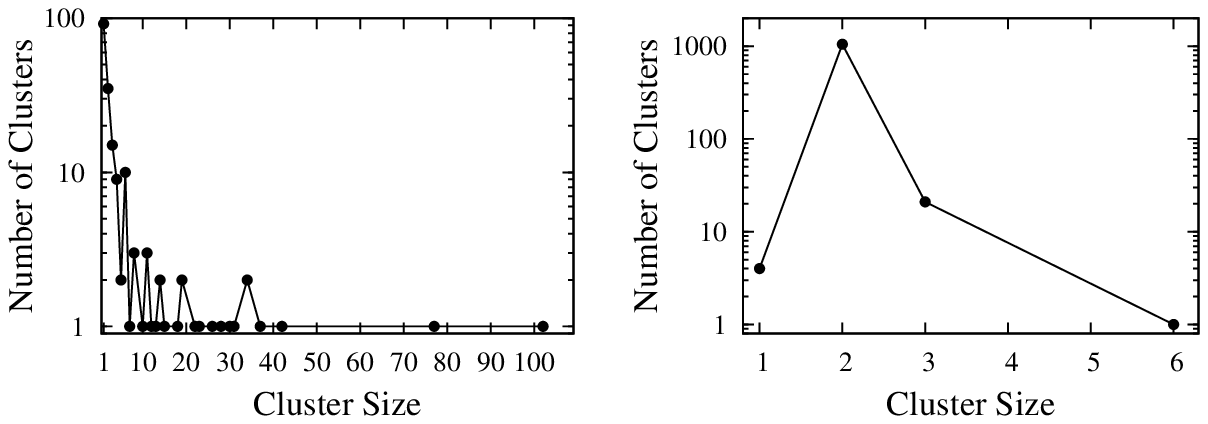}\\
  \hspace{3em} (a) \paper \hspace{9em}   (b) \product \hspace{4em} \vspace{-2em}
  \caption{Cluster-size distribution.}\label{exp:duplicate} \vspace{-1.5em}
\end{figure}

We used two public real-world datasets to evaluate our approaches which were widely adopted by prior works. (a) \paper (a.k.a \cora)\footnote{\scriptsize http://www.cs.umass.edu/$\thicksim$mccallum/data/cora-refs.tar.gz} is a dataset of research publications. Each object in the dataset is a record with five attributes, \emph{Author}, \emph{Title}, \emph{Venue}, \emph{Date} and \emph{Pages}. There are 997 distinct records, leading to $\frac{997*996}{2} = 496,506$ pairs. (b) \product (a.k.a \abtbuy)\footnote{\scriptsize http://dbs.uni-leipzig.de/file/Abt-Buy.zip} is a product dataset containing information on 1081 products from \emph{abt.com} and 1092 products from \emph{buy.com}. Each object is a product record with two attributes, \emph{name} and \emph{price}. The dataset contains a total of $1081 * 1092 = 1,180,452$ pairs.


We chose \paper and \product datasets in the experiment due to their different characteristics in the number of matching objects. To visualize the difference, we clustered the true matching objects in each dataset, and plotted the cluster-size distribution in Figure~\ref{exp:duplicate}. We see that compared to \product, \paper has far larger clusters and should thus benefit more from using transitive relations. For example, there is a cluster consisting of 102 matching objects on the \paper dataset. For such a large cluster, using transitive relations can reduce the number of crowdsourced pairs from $\frac{102*101}{2} = 5151$ to $101$. However, for smaller clusters, e.g. cluster size = 3, using transitive relations can only reduce the number of crowdsourced pairs from $\frac{3*2}{2} = 3$ to $2$.

It was found that most of the pairs in the datasets look very dissimilar, and can easily be weeded out by algorithmic methods~\cite{journals/pvldb/WangKFF12}. We followed this method to compute for each pair a likelihood that they are matching, and only asked the crowd workers to label the most likely matching pairs, i.e. those pairs whose likelihood is above a specified threshold.

\subsection{Effectiveness of Transitive Relations}

\begin{figure}[tbp] 
\centering
  \includegraphics[scale=0.7]{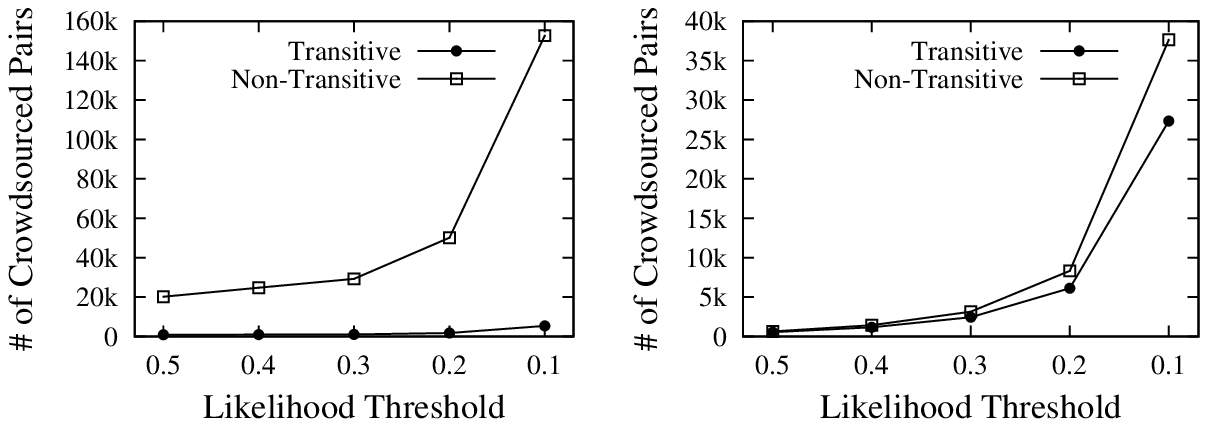}\\
  \hspace{3em} (a) \paper \hspace{9em}   (b) \product \hspace{4em} \vspace{-2em}
  \caption{Effectiveness of transitive relations.}\label{exp:transitivity} \vspace{-1em}
\end{figure}

In this section, we evaluate the effectiveness of transitive relations on the \paper and \product datasets. If the labeling method did not apply transitive relations, denoted by \nontransitive, all of the pairs need to be crowdsourced. On the contrary, if the labeling method utilized transitive relations, denoted by \transitive, many pairs can be deduced based on transitive relations, and only the remaining pairs need to be crowdsourced. We varied the likelihood threshold from 0.5 to 0.1 on both the \paper and \product datasets, and respectively used \nontransitive and \transitive to label the pairs whose likelihood is above the threshold. Figure~\ref{exp:transitivity} compares the number of crowdsourced pairs required by \nontransitive and \transitive  (with the optimal labeling order). On the \paper dataset, we can see \transitive reduced the number of crowdsourced pairs by 95\%. For example, when the likelihood threshold was 0.3, \transitive only needed to crowdsource 1065 pairs while \nontransitive had to crowdsource 29,281 pairs. On the \product dataset, even if there are not so many matching objects in the dataset (Figure~\ref{exp:duplicate}), \transitive can still save about 20\% crowdsourced pairs compared to \nontransitive. For example, when the threshold is 0.2, 6134 pairs needed to be crowdsourced by \transitive while \nontransitive required to crowdsource 8315 pairs.

\subsection{Evaluating Different Labeling Orders}\label{subsec:evaluate-order}
Having shown the benefits of transitive relations in reducing the crowdsourced pairs, we now turn to examining how different labeling orders affect the effectiveness of transitive relations. We compare the number of crowdsourced pairs required by different labeling orders in Figure~\ref{exp:labeling-order}. \optimalorder, \expectoptimalorder, \randomorder, and \worstorder respectively denote the labeling orders which label first all matching pairs then the other non-matching pairs, label the pairs in the decreasing order of likelihood, label the pairs randomly, and label first all non-matching pairs then the other matching pairs. By comparing \worstorder with \optimalorder, we can see the selection of labeling orders has a significant effect on the number of required crowdsourced pairs. For example, on the \paper dataset, if labeling the pairs whose likelihood is above 0.1 in the worst order, we needed to crowdsource 139,181 pairs, which was about 26 times more than the crowdsourced pairs required by the optimal order. By comparing \expectoptimalorder and \randomorder with \optimalorder, we can see that the \expectoptimalorder needed to crowdsource a few more pairs than the \optimalorder but the \randomorder involved much more crowdsourced pairs, which validated that our heuristic labeling order has a very good performance in practice. Unless otherwise stated, we will use the \expectoptimalorder to label the pairs in later experiments.

\begin{figure}[tbp] 
\centering
  \includegraphics[scale=0.7]{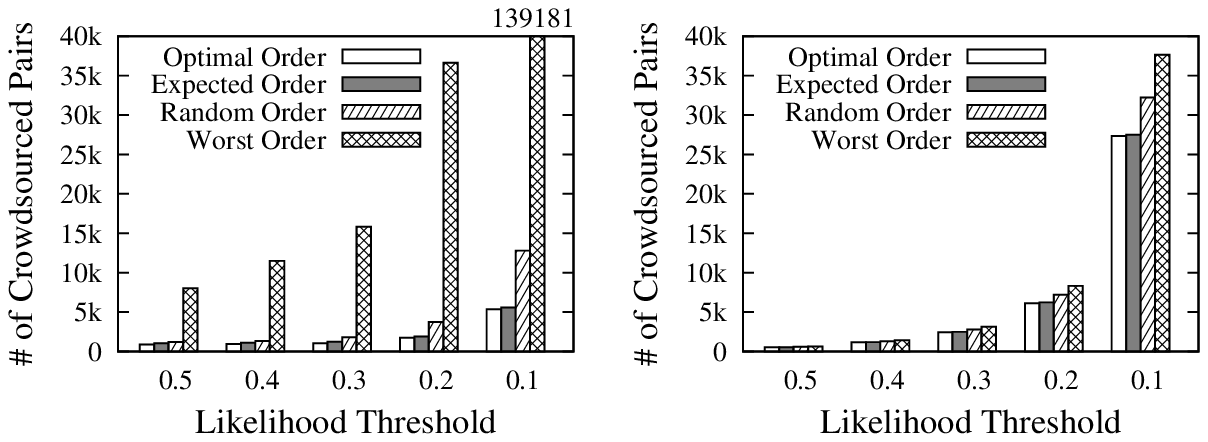}\\
  \hspace{3em} (a) \paper \hspace{9em}   (b) \product \hspace{4em}\vspace{-2em}
  \caption{The number of crowdsourced pairs required by different labeling orders.}\label{exp:labeling-order}\vspace{-1em}
\end{figure}

\subsection{Evaluating Parallel Labeling Algorithm}
In this section, we evaluate our parallel labeling algorithm as well as the corresponding optimization techniques.

We first compare the parallel labeling algorithm (referred to as \parallelno) with the non-parallel labeling algorithm (referred to as \nonparallel). We respectively used \parallelno and \nonparallel to label the pairs whose likelihood was above 0.3.  Figure~\ref{exp:parallel-03} illustrates their number of parallel pairs in each iteration. Compared to \nonparallel, \parallelno significantly reduced the total number of iterations. For example, on the \paper dataset, there were a total of 1237 crowdsourced pairs. For this, Non-Parallel required 1237 iterations, i.e., in each iteration only a single pair could be crowdsourced. But \parallelno reduced the number of iterations to 14, where in each iteration, 908, 163, 40, 32, 20, 18, 11, 9, 9, 9, 7, 6, 4, and 1 pair(s) respectively have been crowdsourced in parallel. We also evaluated \parallelno for other likelihood thresholds, and found that a better performance can be achieved for higher likelihood thresholds. For example, Figure~\ref{exp:parallel-04} shows the result for a threshold of 0.4. Comparing to the result in Figure~\ref{exp:parallel-03} (threshold=0.3), \parallelno involved fewer iterations on both datasets. This is because for a larger threshold, there were fewer number of pairs whose likelihood was above the threshold. Thus the graph built for the pairs became more sparse, and allowed to crowdsource more pairs per iteration.

\begin{figure}[tbp] 
\centering
  \includegraphics[scale=0.7]{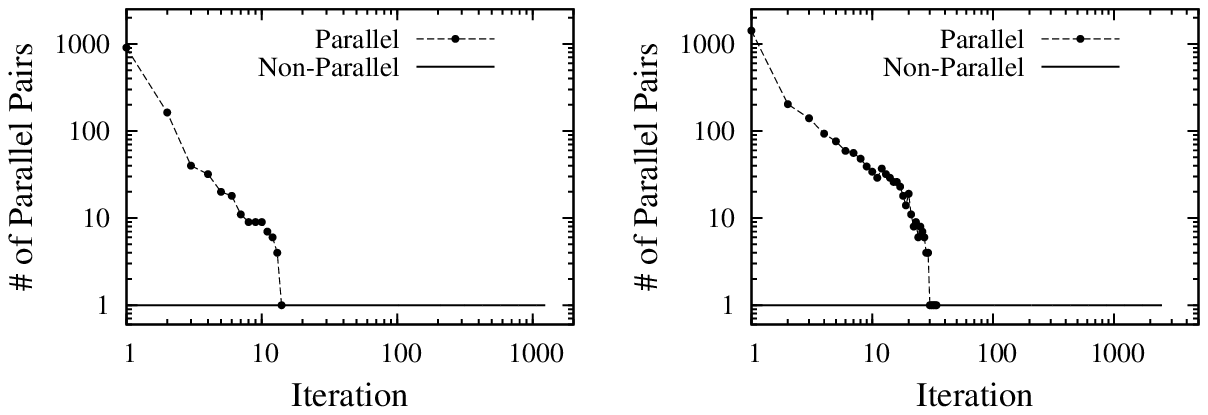}\\
  \hspace{3em} (a) \paper \hspace{9em}   (b) \product \hspace{4em}\vspace{-2em}
  \caption{Parallel v.s. non-parallel labeling algorithm (likelihood threshold = 0.3).}\label{exp:parallel-03}\vspace{1em}
    \includegraphics[scale=0.7]{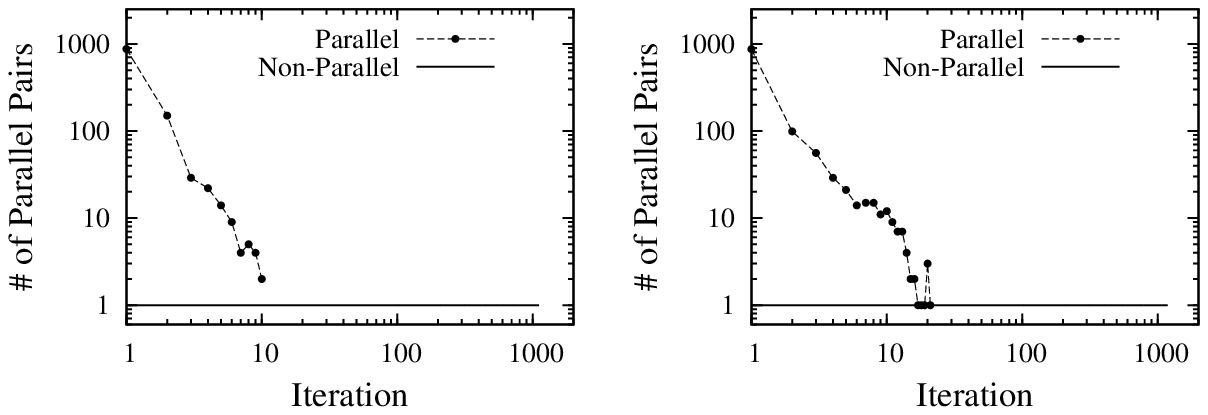}\\
  \hspace{3em} (a) \paper \hspace{9em}   (b) \product \hspace{4em}\vspace{-2em}
  \caption{Parallel v.s. non-parallel labeling algorithm (likelihood threshold = 0.4).}\label{exp:parallel-04} \vspace{-1em}
\end{figure}

Next we evaluate optimization techniques for parallel labeling algorithm. In Figure~\ref{exp:parallel-op-03}, \parallelno, \parallelid, and \parallelnf respectively denote the parallel algorithm without any optimization technique, the parallel algorithm with the instant-decision optimization technique, and the parallel algorithm with both instant-decision and non-matching-first optimization techniques. (Note that the parallel algorithm with only the non-matching-first optimization technique is the same as \parallelno.) At the beginning, all of the three algorithms published a set of pairs to the crowdsourcing platform, and then waited for the crowd workers to label them. \parallelno and \parallelid were supposed to label the pairs randomly while \parallelnf was supposed to first label the most unlikely matching pairs. When a pair was labeled, \parallelno would not publish any new pairs until all the pairs in the crowdsourcing platform had been labeled, whereas both \parallelid and \parallelnf would instantly decide which pair to publish next. Figure~\ref{exp:parallel-op-03} illustrates that the number of available pairs in the crowdsourcing platform changed with the increasing number of pairs labeled by the crowd. Unlike \parallelno, \parallelid and \parallelnf ensured that at any time, there were sufficient pairs available in the crowdsourcing platform, which kept the crowd doing our work continuously. For example, on the \product dataset, after 1420 pairs were crowdsourced, \parallelno only had one available pair in the crowdsourcing platform while \parallelid and \parallelnf respectively had 219 pairs and 281 pairs in the crowdsourcing platform. In addition, we can also see from the figure that \parallelnf lead to more available pairs than \parallelid, which validated the effectiveness of the non-matching-first optimization technique.

\subsection{Evaluating our approaches in a real crowdsourcing platform}

Finally we evaluate our approaches with AMT. We paid workers 2 cents for completing each HIT. In order to reduce the cost, we adopted a batching strategy~\cite{journals/pvldb/MarcusWKMM11,journals/pvldb/WangKFF12} by placing 20 pairs into one HIT. To control the result quality, each HIT was replicated into three assignments. That is, each pair would be labeled by three different workers. The final decision for each pair was made by majority vote.

We first compare \parallelid with \nonparallel in AMT using a threshold of 0.3. (We were unable to evaluate \parallelnf in AMT since the current AMT can only randomly assign HITs to workers.) As we only focused on the difference between their completion time, we simulated that the crowd in AMT always gave us correct labels. In this way, the two algorithms would crowdsource the same number of pairs, thus requiring the same amount of money. As discussed in Section~\ref{subsec:labeling}, the batching strategy is not applicable to \nonparallel. To make a fair comparison, \nonparallel used the same HITs as \parallelid, but published a single one per iteration. Table~\ref{exp:parallel-amt} compares their completion time. We can see \parallelid significantly improved the labeling performance over \nonparallel. For example, on the \paper dataset, if we used \nonparallel to publish 68 HITs in a non-parallel way, we had to wait for 78 hours. However, if we published them in parallel, the waiting time reduced by almost one order of magnitude.

\begin{figure}[tbp] 
\centering
  \includegraphics[scale=0.7]{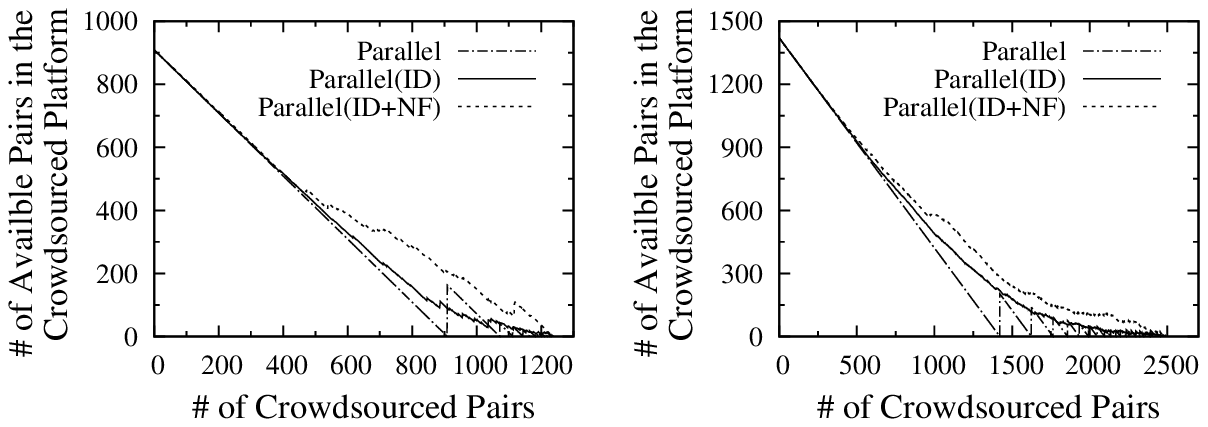}\\
  \hspace{3em} (a) \paper \hspace{9em}   (b) \product \hspace{4em}\vspace{-2em}
  \caption{Optimization techniques for the parallel labeling algorithm (likelihood threshold = 0.3).}\label{exp:parallel-op-03} \vspace{-1.5em}
\end{figure}



\begin{table}[htbp] \vspace{-.5em}
  \centering
  \caption{Comparing \parallelid  with \nonparallel in AMT (likelihood threshold = 0.3).}\label{exp:parallel-amt}\vspace{.5em}
   \begin{tabular}{|@{\,}c||c|c|c@{\,}|}
  \hline
    \bf Dataset    & \bf \# of HITs     &  \bf Non-Parallel  &  \bf Parallel(ID)   \\ \hline\hline
    \paper     & 68            & 78 hours             &  8 hours    \\ \hline
    \product   & 144            & 97 hours             &  14 hours     \\ \hline
  \end{tabular}\vspace{-.5em}
\end{table}

In order to evaluate the effectiveness of transitive relations with AMT, we respectively used \transitive and \nontransitive to label the pairs whose likelihood was above 0.3, where \transitive adopted \parallelid to label the pairs in the \expectoptimalorder, and \nontransitive simply published all the pairs simultaneously to the crowdsourcing platform. We compared \transitive with \nontransitive in terms of completion time, number of HITs, and result quality. Table~\ref{exp:transitive-amt} shows the respective results on the \paper and \product datasets. We employed \emph{Precision}, \emph{Recall}, and \emph{F-measure} to evaluate the result quality. Let \tp denote the number of correctly labeled matching pairs, \fp the number of wrongly labeled matching pairs, and \fn the number of falsely labeled non-matching pairs. Precision and recall are respectively defined as $\frac{\tp}{\tp+\fp}$ and $\frac{\tp}{\tp+\fn}$. F-measure is defined as the harmonic mean of precision and recall, i.e. $\frac{2\cdot \precision\cdot \recall}{\precision+\recall}$.

We used qualification tests to improve the result quality. A qualification test consisted of three specified pairs. Only the workers who correctly labeled the pairs were allowed to do our HITs. In Table~\ref{exp:transitive-amt}(a), to label 29,281 pairs on \paper dataset, \nontransitive published $\frac{29281}{20} = 1465$ HITs, and waited for 755 hours until all the HITs were completed, whereas \transitive can reduce the HITs by 96.5\% and the time by 95.8\% with about 5\% loss in the result quality. This experimental result indicates that for the dataset with a lot of matching objects, \transitive can save a large amount of cost and time with a little loss in result quality. The reason for the loss of quality is that some pairs' labels were falsely deduced from incorrectly labeled pairs based on transitive relations.

\begin{table}[tbp] \vspace{-.5em}
  \centering
\caption{Comparing \transitive with \nontransitive in AMT (likelihood threshold = 0.3)}\label{exp:transitive-amt}\vspace{.0em}
  (a) \paper \scriptsize
   \begin{tabular}{|@{\;}l@{\;}||@{\;}c@{\;}|@{\;}c@{\;}|c@{\;}|@{\;}c@{\;}|@{\;}c@{\;}|}
  \hline
                      & \bf \# of     &  \bf \multirow{2}{*}{Time}   &  \multicolumn{3}{@{\;}c|}{\bf  Quality}   \\ \cline{4-6}
                      &    \bf HITs            &                  &  \bf Precision   & \bf Recall & \bf F-measure \\ \hline
    \nontransitive    & 1465           & 755 hours        &    68.82\%      & 95.03\%  &  79.83\%    \\ \hline
    \transitive        & 52             & 32  hours        &    62.96\%      & 90.47\%  &  74.25\%     \\ \hline
  \end{tabular}\vspace{.5em}\\
    {\normalsize(b)  \product}
   \begin{tabular}{|@{\;}l@{\;}||@{\;}c@{\;}|@{\;}c@{\;}|c@{\;}|@{\;}c@{\;}|@{\;}c@{\;}|}
  \hline
                      & \bf \# of     &  \multirow{2}{*}{\bf Time}   &  \multicolumn{3}{@{\;}c|}{\bf Quality}   \\ \cline{4-6}
                      &  \bf  HITs            &                  &  \bf Precision   & \bf Recall & \bf F-measure \\ \hline
    \nontransitive    & 158             & 22 hours        &    95.69\%      & 68.94\%  &  80.14\%    \\ \hline
    \transitive        & 144             & 30  hours        &    94.70\%      & 68.82\%  &  79.71\%     \\ \hline
  \end{tabular}\vspace{-1.5em}
\end{table}

Next we turn to the experimental result on \product dataset. In Table~\ref{exp:transitive-amt}(b), 3154 pairs needed to be labeled, and \nontransitive published $\frac{3154}{20} = 158$ HITs, and waited for 22 hours until all the HITs were completed. Since there are not so many matching objects in the dataset, \transitive can only save about 10\% of the HITs. Due to the iterative process of publishing HITs, \transitive lead to a little longer completion time. But in terms of quality, \transitive was almost the same as \nontransitive. This experimental result indicates that for the dataset with not so many matching objects, transitive relations can help to save some money with almost no loss in result quality but may lead to longer completion time.

\vspace{-.5em}
\section{Related Work}
\label{sec:related-work}
Recently, several projects on crowd-enabled query processing system~\cite{conf/sigmod/FranklinKKRX11,journals/pvldb/MarcusWKMM11,ilprints1015} and hybrid crowd-machine data integration system~\cite{jeffery2013arnold} were proposed in the database community. To implement such systems, there are many studies in processing a variety of crowdsourced queries~\cite{journals/pvldb/MarcusWKMM11,journals/pvldb/WangKFF12,conf/www/DemartiniDC12,ilprints1047,journals/pvldb/ParameswaranSGPW11,conf/sigmod/GuoPG12,conf/sigmod/ParameswaranGPPRW12,conf/www/VenetisGHP12,conf/icde/BethTMP13}. As one of the most important queries, crowdsourced joins have been widely investigated in~\cite{journals/pvldb/MarcusWKMM11,journals/pvldb/WangKFF12,conf/www/DemartiniDC12,ilprints1047}. Marcus et al.~\cite{journals/pvldb/MarcusWKMM11} proposed a human-only technique with some batching and feature filtering optimizations for crowdsourced joins. Wang et al.~\cite{journals/pvldb/WangKFF12} developed a hybrid human-machine workflow which first utilized machine-based techniques to weed out a large number of obvious non-matching pairs, and only asked the crowd workers to label the remaining pairs. Demartini et al.~\cite{conf/www/DemartiniDC12} also employed a hybrid human-machine technique, and in addition, they developed a probabilistic framework for deriving the final join result. Whang et al.~\cite{ilprints1047} proposed a budget-based method for crowdsourced joins which assumed there was not enough money to label all the pairs, and explored how to make a good use of limited money to label a certain number of pairs. When using the crowd workers to label a set of pairs, the prior works neglected the fact that transitive relations hold among the pairs. Therefore, our work complements them by leveraging transitive relations to reduce the number of crowdsourced pairs. In a recent technical report, Gruenheid et al.~\cite{tr785} also explored how to leverage transitive relations for crowdsourced joins, but they mainly studied how to decide whether two objects refer to the same entity when crowd workers give inconsistent answers, which has a different focus than our work.


Some real applications such as entity resolution~\cite{journals/ml/BansalBC04,conf/sigmod/WhangMKTG09,conf/dmkd/MongeE97} also seek to benefit from transitive relations. Essentially, these works first label pairs using some sophisticated algorithms, and then utilize transitive relations to obtain the final result. They mainly focused on how to resolve the conflicts introduced by transitive relations rather than reduce the precious human effort. In a pay-as-you-go data integration system, Jeffrey et al.~\cite{conf/sigmod/JefferyFH08} studied the problem of minimizing the human work to achieve the best data-integration quality. But their approach aimed to identify the most uncertain pairs for verification, without considering the benefits of transitive relations.

There are also some studies on parallel crowdsourcing. Little et al.~\cite{little2010exploring} compared iterative crowdsourcing model with parallel crowdsourcing model in a variety of problem domains, and provided some advice about the selection of models. TurKit~\cite{conf/uist/LittleCGM10} is a toolkit based on the crash-and-rerun programming model which makes it easier to write parallel crowdsourcing algorithms. CrowdForge~\cite{conf/uist/KitturSKK11} is a map-reduce style framework which partitions a complex task into subtasks that can be done in parallel. These tools can help us to easily implement the parallel labeling algorithm in a real crowdsourcing platform.

\vspace{-.8em}
\section{Conclusion and Future Work}
\label{sec:conclusion}
\vspace{-.3em}

We studied the problem of leveraging transitive relations for crowdsourced joins (e.g., for entity resolution), to minimize the number of crowdsourced pair verifications.  Our approach consists of two components:
(1) The sorting component takes the pre-matched pairs from a machine-based method and determines the best order for verification. We found that the labeling order has a significant effect on the number of crowdsourced pairs. We proved that the optimal labeling order, which minimizes the number of crowdsourced pairs, has to first label all the matching pairs, and then label the other non-matching pairs. As this is impossible to achieve (we do not know the real matching pairs upfront), we proposed a heuristic labeling order that labels the pairs in the decreasing order of the probability that they are a matching pair. (2) For the labeling component, we found that a simple labeling method lead to longer latency and increased cost. We devised a novel parallel labeling algorithm to overcome these drawbacks. We have evaluated our approaches, both using simulation and with AMT, and showed, that transitive relations can lead to significant cost savings with no or little loss in result quality.

\sloppy

Various future directions exist for this work including to detect non-transitive relations, automate money/time/quality trade-offs for joins, explore other kinds of relations (e.g. one-to-one relationship) or extend to non-equality joins (i.e., general theta-joins). In this work, we showed that transitive relations can lead to significant cost savings in crowdsourced joins for entity resolution.

\fussy

\vspace{.5em}

\fussy
{\noindent  \bf Acknowledgements.} {\small This work was partly supported by the National Natural Science Foundation of China under Grant No.~61003004 and 61272090, National Grand Fundamental Research 973 Program of China under Grant No.~2011CB302206, and a project of Tsinghua University under Grant No. 20111081073, and the ``NExT Research Center'' funded by MDA, Singapore, under Grant No.~WBS:R-252-300-001-490, and NSF CISE Expeditions award CCF-1139158 and DARPA XData Award FA8750-12-2-0331, and  gifts from Amazon Web Services, Google, SAP,  Blue Goji, Cisco, Clearstory Data, Cloudera, Ericsson, Facebook, General Electric, Hortonworks, Huawei, Intel, Microsoft, NetApp, Oracle, Quanta, Samsung, Splunk, VMware and Yahoo!. 
}
\sloppy

\vspace{-.5em}

{
\bibliographystyle{abbrv}
\scriptsize
\bibliography{ref/crowdsourcing,ref/er}
}
\end{document}